\documentclass[prx,aps,10pt,twocolumn,superscriptaddress,floatfix,nofootinbib,longbibliography]{revtex4-2}
\pdfoutput=1 
\usepackage{amsmath,amssymb,amsfonts,amsthm}
\usepackage{graphicx}
\usepackage{enumitem}
\usepackage{bm}
\usepackage{rotating}
\usepackage[colorlinks=true,linkcolor=blue, citecolor=blue, urlcolor=blue, bookmarks]{hyperref}
\usepackage{pbox}
\usepackage{array}
\usepackage{mathtools}
\usepackage{leftidx}
\usepackage{lmodern}
\usepackage[normalem]{ulem}
\usepackage{braket}
\usepackage{tensor}
\usepackage[usenames,dvipsnames]{xcolor}
\usepackage{float}
\usepackage{footnote}
\usepackage{dsfont}
\usepackage{mathrsfs}
\usepackage{bbm}
\usepackage{tikz}
\usepackage{comment}
\usepackage[many]{tcolorbox}
\usepackage{booktabs}

\usetikzlibrary{calc,matrix,positioning,shapes}

\DeclareMathOperator{\diag}{diag}

\newcommand{\R}{\mathbb{R}}
\newcommand{\C}{\mathbb{C}}

\newcommand{\dd}{\mathrm{d}}
\newcommand{\rr}[1]{\left(#1\right)}

\DeclareMathOperator{\Tr}{Tr}
\renewcommand{\tilde}{\widetilde}
\renewcommand{\bar}{\overline}

\newcommand{\N}{\mathbb{N}}
\renewcommand{\Re}{\mathrm{Re}}
\renewcommand{\Im}{\mathrm{Im}}
\renewcommand{\openone}{\mathbbm{1}}

\newcommand{\ketbra}[2]{{\left| {#1} \right\rangle \!\!\left\langle {#2} \right|}}

\DeclareMathOperator{\rank}{rank}

\DeclareMathOperator{\im}{im}

\DeclareMathOperator{\Span}{span}
\newtheorem{proposition}{Proposition}

\newtheorem{lemma}{Lemma}
\newtheorem{definition}{Definition}
\newtheorem{corollary}{Corollary}
\theoremstyle{definition}
\newtheorem{example}{Example}
\theoremstyle{definition}

\newcommand{\ordprod}{%
    \mathop{\overrightarrow{\prod}}\displaylimits
}

\newcommand{\exampleqed}{\hfill{$\blacklozenge$}}
\makeatother

\begin{document}

\title{Gaussian-augmented bosonic matrix-product states: theory and applications}

\author{Erickson Tjoa}
\affiliation{Max-Planck-Institut f\"ur Quantenoptik, Hans-Kopfermann-Stra\ss e 1, D-85748 Garching, Germany}
\affiliation{Munich Center for Quantum Science and Technology (MCQST), Schellingstraße 4, D-80799 Munich, Germany}

\author{J. Ignacio Cirac}
\affiliation{Max-Planck-Institut f\"ur Quantenoptik, Hans-Kopfermann-Stra\ss e 1, D-85748 Garching, Germany}
\affiliation{Munich Center for Quantum Science and Technology (MCQST), Schellingstraße 4, D-80799 Munich, Germany}

\begin{abstract}

    We propose and analyze the structure of a family of bosonic quantum many-body states that have the following features: (i) they include all pure Gaussian states and finite-dimensional matrix-product states as subclasses; (ii) their expectation values can be efficiently computed, allowing them to be used, among other things, for variational calculations; (iii) they admit exact parent Hamiltonians expressed as simple functions of the bosonic creation and annihilation operators.

\end{abstract}

\maketitle

\section{Introduction}

Many-body bosonic systems appear naturally throughout physics, from condensed-matter systems and ultracold atomic gases to quantum optics and quantum field theory \cite{Bloch2008ultracold,Jaksch1998OpticalLattices,Fisher1989BosonLocalization,Cazilla2011bosons1D,lewenstein2012ultracold,Treps2020quantumoptics,Plenio2008cavityarrays,GlimmJaffee1968phi4-1}. One of the main challenges in modern physics is understanding the phenomena of strongly-interacting quantum many-body systems beyond the perturbative regimes. Thus there are extensive efforts in developing methods to tackle both static and dynamical problems involving strongly-interacting bosonic systems. 

Two standard approaches to address these questions are Monte Carlo methods \cite{Ceperley1995MonteCarlo,DuBois2001variationalMonteCarlo,Flottat2015bosonicKondo,Bosetti2015montecarlophi4,Bronzin2015montecarlophi4} and variational wavefunctions \cite{gross1961structure,pitaevskii1961vortex,Pitaevskii1999,pethick2008BEC,Rokshar1991gutzwiller,Shi2018NonGaussianStates,Guaiata2019gaussianTDVP,hackl2020geometry,Qian2023NonGaussianBosons}. Among the variational methods, a widely used class is based on Gaussian states \cite{Weedbrook2012GaussianQI}: they are pure states with Gaussian wavefunctions and can be efficiently parametrized by $O(N^2)$ parameters for $N$ modes. As a variational family, Gaussian states can be used to describe Bose–Einstein condensation \cite{gross1961structure,pitaevskii1961vortex,Pitaevskii1999,pethick2008BEC}, quasiparticle excitations, and (non-)equilibrium dynamics \cite{Guaiata2019gaussianTDVP}. However, their efficient description comes at a cost of having very constrained correlations: due to Wick's theorem, Gaussian states are completely characterized by their one- and two-point correlation functions. Consequently, some non-Gaussian variational families have been introduced by applying certain types of non-Gaussian transformations on Gaussian states \cite{Shi2018NonGaussianStates,Qian2023NonGaussianBosons}. 

Another complementary variational approach is provided by tensor networks  \cite{Cirac2021TensorNetworks,PerezGarcia2007MPSRepresentations,Verstraete2008MPS,Orus2014TensorNetworks,Bridgeman2017TensorNetworks}. In particular, in one-dimensional settings the so-called matrix-product states (MPS) provide state-of-the-art variational methods through the density-matrix renormalization group (DMRG) algorithm and related algorithms \cite{White1992DMRG,Verstraete2004dmrg,Schollwoeck2011DMRGReview,Vidal2004TEBD}. However, for bosonic systems these methods require truncation of the local Hilbert space that may limit its applicability and obscure the physics under study. Several ideas have been put forward in recent years to improve the numerical methods, e.g., by using local basis optimization \cite{Jeckelmann1998HolsteinPolaron,Zhang1998dmrg,Brockt2015LocalBasis,Stolpp2021BosonicLocalBasis} or displacing the Hamiltonian before truncation \cite{Guo2012displaced}. Several proposals adapt a combination of Gaussian and MPS frameworks \cite{Iblisdir2007mps,Schuch2008GaussianMPS,frenzel2013matrix,Michelsen2025FunctionalMPS,janik2019exact}. 

Besides their usefulness as a computational tool, MPS also provide an important analytical framework for quantum many-body physics. Many physically relevant lattice models are translationally invariant, and a translationally invariant MPS is fully specified by a single local tensor and possibly some boundary data for finite chains. Within the MPS framework, properties of the local tensor can be used to classify one-dimensional gapped phases, including symmetry-breaking and symmetry-protected topological phases \cite{Chen2011CompleteClassification,Schuch2011ClassifyingPhases,Pollmann2012SymmetryProtection}. Furthermore, for a given MPS one can systematically associate a local frustration-free parent Hamiltonian for which the MPS is an exact ground state. This construction provides physical models whose ground-state properties can be characterized directly from the tensors \cite{PerezGarcia2007MPSRepresentations,Cirac2021TensorNetworks}.

In this paper we propose and analyze the structure of a rich family of bosonic many-body states called Gaussian-augmented bosonic matrix-product states (GA-BMPS) that have the following features: (i) they include all pure Gaussian states and truncated (finite-dimensional) MPS; (ii) expectation values can be efficiently computed and hence, among other things, they can be used for variational calculations; (iii) they have exact parent Hamiltonians that are simple functions of the canonical operators. Our work unifies previous constructions \cite{PerezGarcia2007MPSRepresentations,Schuch2008GaussianMPS,Weedbrook2012GaussianQI,frenzel2013matrix} as special cases, with its non-Gaussian expressivity controlled by the bond dimension.

This paper is organized as follows. In Section~\ref{sec: GA-BMPS} we review the required notions of Gaussian states and finite-dimensional MPS and introduce the ansatz. In Section~\ref{sec: structure} we analyze its structure, expressivity, transfer-matrix calculus, alternative representations, gauge freedom, and natural extensions. In Section~\ref{sec: parent-hamiltonian} we construct parent Hamiltonians for these states. In Section~\ref{sec: variational} we present numerical tests of their viability as a variational family. We conclude in Section~\ref{sec: discussions} with extensions and open directions.

\section{Gaussian-augmented bosonic MPS}
\label{sec: GA-BMPS}

We consider a 1D lattice of $N$ sites with one bosonic mode per site. The generalizations to include multiple modes per site and/or spins are straightforward (Sec.~\ref{subsec: extensions}). The local Hilbert space is described by the Fock space
\begin{align*}
    \mathcal{F}
    =
    \overline{\Span_\C\{\ket{n}:n\in\N_0\}}
    \cong
    L^2(\R)\,,
\end{align*}
spanned by orthonormal basis vectors $\ket{n}$ (Fock states) $\ket{n}=\tfrac{1}{\sqrt{n!}}(a^\dagger)^n\ket{0}$. The canonical operators at each site $j$ are given by the annihilation and creation operators (collectively called ladder operators) $a_j,a_j^\dagger$ satisfying the canonical commutation relation (CCR) $[a_i,a_j^\dagger] = \delta_{ij}\openone$. The quadrature operators are related to the ladder operators as $x_j\coloneqq \frac{1}{\sqrt{2}}(a_j+a^\dagger_j)$ and $ p_j\coloneqq\frac{i}{\sqrt{2}}(a_j^\dagger-a_j)$. 

\subsection{Gaussian states}

Let $\mathbf a=(a_1,\ldots,a_N)^T$ be the vector of annihilation
operators. Gaussian unitaries are generated by Hamiltonians that are
at most quadratic in the bosonic creation and annihilation operators.
For a time-independent generator, we write
\begin{equation}
    \begin{aligned}
        U_{\mathsf G}
        &=e^{i\chi}e^{-iH_{\mathsf G}},
        \qquad \chi\in\R,\\
        H_{\mathsf G}
        &=
        \mathbf a^\dagger h\mathbf a
        +
        \frac{1}{2}\left(
            \mathbf a^\dagger\Delta\mathbf a^{\dagger T}
            +
            \mathbf a^T\bar{\Delta}\mathbf a
        \right)
        +
        \mathbf a^\dagger\boldsymbol{\eta}
        +
        \boldsymbol{\eta}^\dagger\mathbf a,
    \end{aligned}
\end{equation}
where $h=h^\dagger$, $\Delta=\Delta^T$, and
$\boldsymbol{\eta}\in\C^N$. Defining the quadratures $\mathbf{R} \coloneqq (x_1...,x_N,p_1,...,p_N)^T$, the CCR takes the form $[R_i,R_j]=i\Omega_{ij}\openone$ where $\Omega$ is the symplectic form
\begin{align}
    \Omega = \begin{bmatrix}
        0 & \openone_N \\ -\openone_N & 0
    \end{bmatrix}\,.
\end{align}
In the Heisenberg picture, $U_\mathsf{G}$ implements affine symplectic transformation on $\mathbf{R}$ \cite{Weedbrook2012GaussianQI},
\begin{align}
    U_{\mathsf G}^\dagger\mathbf{R} U_{\mathsf G}
    =
    S\mathbf{R}+\mathbf{d}\,,
    \label{eq: Gaussian-affine-map}
\end{align}
where $S\in Sp(2N,\R)$ are symplectic matrices (they satisfy $S\Omega S^T=\Omega$) and $\mathbf{d}\in \R^{2N}$. 

A pure Gaussian state is, up to global phase, given by
\begin{align}
    \ket{\Psi_{\mathsf G}}= U_{\mathsf{G}}\ket{0}^{\otimes N}\,,
    \label{eq: Gaussian-states}
\end{align}
where $\ket{0}$ is the Fock vacuum.

\subsubsection{Single-mode Gaussian states} 

Two particularly important families of single-mode Gaussian states are \textit{coherent states} and \textit{squeezed states},
\begin{equation}
    \begin{aligned}
        \ket{\alpha}  &= D(\alpha)\ket{0}\,,\quad D(\alpha) = e^{\alpha a^\dagger - \alpha^*a}\,,\\
        \ket{\zeta} &= S(\zeta)\ket{0}\,,\,\quad S(\zeta) =  e^{\frac{1}{2}\rr{\zeta^* a^2 - \zeta(a^\dagger)^2}}\,,
    \end{aligned}
\end{equation}
defined for $\alpha,\zeta\in \C$. Up to a global phase, any single-mode Gaussian unitary can be written as
\begin{align}
    U_\mathsf{G} = D(\alpha)S(\zeta)e^{-i\theta a^\dagger a}
    \label{eq: single-mode-gaussian-unitary}
\end{align}
and hence any pure Gaussian state, up to a global phase, can be written as \textit{squeezed coherent states}
\begin{align}
    \ket{\Psi_{\mathsf{G}}} = D(\alpha)S(\zeta)\ket{0}\,.
\end{align}

For our purposes, it will be useful to define a non-unitary parametrization of squeezed coherent states. By writing $\zeta=re^{i\phi}$, we note that
\begin{align}
    \ket{\kappa,\ell} \coloneqq e^{\kappa (a^\dagger)^2+\ell a^\dagger}\ket{0}\propto D(\alpha)S(\zeta)\ket{0}\,,
    \label{eq: squeezed-coherent-state}
\end{align}
where $\kappa = -\frac{1}{2}e^{i\phi}\tanh r$ (hence $|\kappa|<\frac{1}{2}$) and $\ell = \alpha -2\kappa \alpha^*$. The unnormalized squeezed coherent state $\ket{\kappa,\ell}$ has squared norm
\begin{align}
    \braket{\kappa,\ell|\kappa,\ell} = \frac{1}{\sqrt{1-4|\kappa|^2}}\exp\rr{\frac{|\ell|^2+\kappa \bar{\ell}^2 + \bar{\kappa}\ell^2}{{1-4|\kappa|^2}}}\,.
    \label{eq: normalization-squeezed-coherent-state}
\end{align}

Given single-mode Gaussian states, we can define the so-called \textit{photon-added} Gaussian states
\begin{align}
    \ket{{n,\kappa,\ell}}
    &\coloneqq (a^\dagger)^ne^{\kappa (a^\dagger)^2+\ell a^\dagger}\ket{0} \,,
    \label{eq: photon-added-Gaussian-states}
\end{align}
where the photon addition is also unnormalized for convenience. For fixed $\kappa$ and $\ell$, states with different values of $n$ are generally not orthogonal. When $\kappa=\ell=0$ they reduce to Fock states and they are Gaussian only when $n=0$. 

\subsubsection{Multimode Gaussian states}

A pure multimode Gaussian state is likewise obtained by applying a
Gaussian unitary to the multimode vacuum. By the Bloch-Messiah
decomposition (Appendix~\ref{appendix: Bloch-Messiah}), every such state
can be written, up to a global phase, as
\begin{align}
    \ket{\Psi_{\mathsf G}^N}
    =
    \mathsf{U}
    \bigotimes_{j=1}^N
    \left[
        D_j(\alpha_j)S_j(\zeta_j)\ket{0}_j
    \right],
    \label{eq: bloch-messiah-decomposition-states}
\end{align}
where $\mathsf{U}$ is a passive linear-optical unitary of the form
\begin{align}
    \mathsf{U}
    \coloneqq
    e^{
        i\sum_{j,k}
        \mathsf M_{jk}a_j^\dagger a_k
    },
    \label{eq: PLO}
\end{align}
where $\mathsf M$ is Hermitian. In particular, $\mathsf{U}$ preserves
the total particle number,
\begin{align*}
    \mathsf{U}^\dagger n_{\mathsf{tot}}\mathsf{U}
    &=
    n_{\mathsf{tot}}\,,\quad 
    n_{\mathsf{tot}}
    \coloneqq
    \sum_{j=1}^N a_j^\dagger a_j
\end{align*}

By Wick's theorem, a Gaussian state is completely characterized by its
displacement vector and covariance matrix
\begin{align}
    r_i
    \coloneqq
    \braket{R_i}\,,\quad \Gamma_{ij}
    \coloneqq
    \frac{1}{2}
    \braket{
        \left\{
            R_i-r_i,
            R_j-r_j
        \right\}
    }\,.
\end{align}
Consequently, one can use the real parameters $\{\Gamma_{ij},r_i\}$, rather than $\{\mathsf M_{jk},\alpha_j,\zeta_j\}$, to parametrize the Gaussian
state.

To state translational invariance in terms of lattice sites, we group the
quadratures at each site into
\begin{align}
    \mathbf R^{(j)}
    &\coloneqq
    \begin{pmatrix}
        x_j\\
        p_j
    \end{pmatrix},
    \qquad
    \mathbf r^{(j)}
    \coloneqq
    \braket{\mathbf R^{(j)}}
\end{align}
so that $\mathbf{R}^{(j)}_1=x_j$ and $\mathbf{R}^{(j)}_2 = p_j$, and similarly let $\Gamma^{(j,k)}$ denote the $2\times2$ covariance
block between sites $j$ and $k$, with entries
\begin{align}
    \bigl[\Gamma^{(j,k)}\bigr]_{\mu\nu}
    \coloneqq
    \frac{1}{2}
    \braket{
        \{
            R^{(j)}_\mu-r^{(j)}_\mu,
            R^{(k)}_\nu-r^{(k)}_\nu
        \}
    }\,,
\end{align}
where $\mu,\nu\in\{1,2\}$. For a translationally invariant Gaussian state,
$\mathbf r^{(j)}=\mathbf r$ is independent of the site, while
$\Gamma^{(j,k)}=\gamma(j-k)$ depends only on the lattice separation.

\subsection{Finite-dimensional MPS}
Consider a one-dimensional lattice of $N$ sites with local Hilbert space 
$\C^d$ and orthonormal basis $\{\ket{i}\}_{i=0}^{d-1}$. A matrix product 
state (MPS) with bond dimension $D$ is
\begin{align}
    \ket{\Psi_N(A)}
    =
    \sum_{i_1,\ldots,i_N=0}^{d-1}
    \Tr\rr{BA^{i_1}_1\cdots A^{i_N}_N}
    \ket{i_1\cdots i_N}\,,
    \label{eq: MPS-qudit}
\end{align}
where $A_j^i\in M_D(\C)$ and $B\in M_D(\C)$ is a boundary matrix. We say that the MPS has uniform tensors if $A_j^i=A^i$ for all $j$, and that it has periodic boundary condition (PBC) if $B=\openone$.

For a uniform MPS, expectation values can be computed using the
transfer matrices
\begin{align}
    E
    &=
    \sum_{i=0}^{d-1}
    \bar{A^i}\otimes A^i\,,
    \quad
    E_O
    =
    \sum_{m,n=0}^{d-1}
    \braket{m|O|n}
    \bar{A^m}\otimes A^n\,,
\end{align}
where $O$ is a local operator. For non-uniform tensors, one instead
uses the corresponding site-dependent transfer matrices $E_j$ and
$E_{O,j}$. The new issue in the infinite-dimensional setting is that
the sums over physical indices must be well-defined and explicitly
computable.

We will use the standard notion of injectivity
\cite{PerezGarcia2007MPSRepresentations,Fannes1992FinitelyCorrelated}.
A uniform tensor $A$ is injective after blocking $l$ sites if
\begin{align}
    \Span_\C
    \{A^{i_1}\cdots A^{i_l}\}
    =
    M_D(\C)\,.
\end{align}
For an injective uniform MPS tensor, the standard finite-range parent
Hamiltonian has the periodic MPS as its unique ground state for all
sufficiently large system sizes and is uniformly gapped. On an open
chain, its ground space generally contains additional boundary degrees
of freedom. More generally, after sufficient blocking, a tensor in canonical form decomposes into injective blocks. The corresponding block-injective
parent Hamiltonian has a finite ground-state degeneracy that is independent of the system size for all sufficiently large system sizes.

\subsection{Gaussian-augmented bosonic MPS (GA-BMPS)}
\label{subsec: ansatz}

\begin{definition}[GA-BMPS]
\label{def: CV-MPS}
Let $B,V_j,K_j,L_j\in M_D(\C)$ for all $j=1,2,...,N$. The Gaussian-augmented bosonic MPS ansatz \emph{(GA-BMPS)} is defined to be
\begin{align}
    \ket{\psi_N}
    =
    \mathsf{U}\Tr_D\left[
    B\ordprod_{j=1}^N
    V_j e^{K_j\otimes (a_j^\dagger)^2}e^{L_j\otimes a_j^\dagger}
    \right]\ket{0}^{\otimes N}
    \label{eq: exponential-ansatz}
\end{align}
where $\mathsf{U}$ is the passive linear unitary \eqref{eq: PLO} and the path-ordered product is
\begin{align*}
    &\ordprod_{j=1}^N
    V_j e^{K_j\otimes (a_j^\dagger)^2}
    e^{L_j\otimes a_j^\dagger}
    \notag\\
    &\qquad \coloneqq
    V_1 e^{K_1\otimes (a_1^\dagger)^2}
    e^{L_1\otimes a_1^\dagger}
    \cdots
    V_N e^{K_N\otimes (a_N^\dagger)^2}
    e^{L_N\otimes a_N^\dagger}\,.
\end{align*}
Here $\Tr_D$ denotes the trace over the $D$-dimensional auxiliary space on which $B,V_j,K_j,L_j$ act. We say that 
\begin{enumerate}[label=\emph{(\roman*)},leftmargin=*]
    \item $\ket{\psi_N}$ has uniform tensors if $V_j=V,K_j=K,L_j=L$ for all $j$, and it has periodic boundary condition if $B=\openone$. 
    
    \item $\ket{\psi_N}$ has commuting generators if $[K_j,L_j]=0$ for all $j$. 

    \item $\ket{\psi_N}$ is a bosonic MPS \emph{(BMPS)} if $\mathsf{U}=\openone$. 
\end{enumerate}
\noindent We impose $\rho(K_j)<\frac{1}{2}$ for all $j$ which is sufficient for the state to have finite norm.
\end{definition}
\noindent This state has finite norm since the transfer matrix (Sec.~\ref{sec: structure}) is finite (Appendix~\ref{appendix: ansatz-absolute-convergence}). To reduce notational clutter, in what follows we assume GA-BMPS with uniform tensors unless otherwise stated. We mention that this family contains the subfamily $\mathsf U=\openone$, $K_j=0$, and $V_j=\openone$ in \cite{frenzel2013matrix}.

While the passive linear-optical unitary $\mathsf{U}$ comes from the Gaussian formalism, the BMPS part 
\begin{align}
    \ket{\Psi_N} \coloneqq  \Tr_D\left[
    B\ordprod_{j=1}^N
    V_j e^{K_j\otimes (a_j^\dagger)^2}e^{L_j\otimes a_j^\dagger}
    \right]\ket{0}^{\otimes N}
    \label{eq: exponential-ansatz-BMPS}
\end{align}
can be shown to be an infinite-dimensional uniform MPS: indeed, we can re-express it in the Fock basis 
\begin{align}
    \ket{\Psi_N}
    =
    \sum_{n_1,\ldots,n_N = 0}^\infty
    \Tr\rr{BA^{n_1}\cdots A^{n_N}}
    \ket{n_1\cdots n_N}\,.
    \label{eq: MPS-infinite-physical-index}
\end{align}
where the uniform bulk tensors are
\begin{align}
    A^n = \sqrt{n!}VC_n\,,\qquad C_n= \sum_{m=0}^{\lfloor n/2\rfloor}
    \frac{K^mL^{n-2m}}{m!(n-2m)!}\,.
\end{align}
Thus Def.~\ref{def: CV-MPS} defines a well-defined infinite-dimensional bosonic MPS. For $D=1$, the BMPS part is a product of single-mode Gaussian states, whereas for $D\geq2$ the ansatz can describe non-Gaussian states.

\section{Structure of GA-BMPS}
\label{sec: structure}

\subsection{Connection with Gaussian states and MPS}

The GA-BMPS family of states contains well-known variational families of states as special cases, namely Gaussian states and (Fock-encoded) finite-dimensional MPS.

First, by setting the bond dimension $D=1$ and identifying $K_j=\kappa_j,L_j=\ell_j,V_j=v_j,B=b\in \C$ with $|\kappa_j|<\frac{1}{2}$, we obtain
\begin{align}
    \ket{\psi_N} = b\prod_{j=1}^N v_j\cdot \mathsf{U}\bigotimes_{j=1}^Ne^{\kappa_j (a^\dagger_j)^2+\ell_j a^\dagger_j}\ket{0}^{\otimes N} \,,
\end{align}
that is, the family of multimode Gaussian states with the non-unitary parametrization of the squeezed coherent states at each site $j$. Thus one can view the GA-BMPS as essentially the Bloch-Messiah formulation of Gaussian states but with  matrix-valued coherent amplitude and squeezing parameter. 

Second, we can recover the family of finite-dimensional qudit MPS by embedding them into the BMPS ansatz. Let $\mathcal{A}=\sum_{j=0}^{d-1}\mathcal{A}^j\otimes \ket{j}$ be an MPS tensor where $\mathcal{A}^j\in M_D(\C)$ and let the boundary matrix $\mathcal{B}\in M_D(\C)$. Given the ansatz \eqref{eq: exponential-ansatz}, now set
\begin{equation}
\begin{aligned}
    \mathsf{U} &= \openone\,,\qquad B\coloneqq \mathcal{B}\otimes \openone_d\,,\qquad  K = 0\,,\\
    V&\coloneqq \sum_{j=0}^{d-1}\sqrt{j!}\mathcal{A}^j\otimes\ketbra{0}{j}\,,\quad 
    L\coloneqq \openone_D\otimes \sum_{j=0}^{d-2}\ketbra{j+1}{j}\,.
\end{aligned}
\label{eq: MPS-embedding}
\end{equation}
Since $L$ is nilpotent with nilpotency index $d$, we recover the qudit MPS with bond dimension $D$ as a subfamily of BMPS with bond dimension $Dd$. The bond dimension of the BMPS is necessarily larger in order to embed $d$ matrices $\{\mathcal{A}^j\}$ into a single matrix $V$ and use the $d$-dimensional sector to encode the qudit Hilbert space using the nilpotency of $L$.

\subsection{Expressivity}

The GA-BMPS is expressive in that it can approximate any state in the full $N$-mode Fock space. This is because superposition of multimode coherent states is contained in GA-BMPS by setting $B=\openone, K_j=0$, and $L_j$ diagonal matrices. Since finite linear combinations of coherent states are dense in Fock space, allowing $D$ to increase enables us to approximate any state to arbitrary accuracy.

Furthermore, every GA-BMPS can be represented as a commuting-generator GA-BMPS with larger bond dimension. Hence, by allowing $D$ to increase, the commuting-generator family is as expressive as the full family. For simplicity, we illustrate the construction for PBC ($B=\openone$) and diagonalizable $K,L$; the general argument is given in Sec.~\ref{subsec: alternative-forms}. Note that we are only concerned with the BMPS part and the passive unitary $\mathsf{U}$ is unaffected by this embedding.

For simplicity assume that we have uniform tensors. Starting from the state $\ket{\Psi_N}$ in \eqref{eq: exponential-ansatz-BMPS}, we first diagonalize $K,L$ to obtain
\begin{align*}
    K &= XD^{[K]}X^{-1}\,,\quad D^{[K]}=\diag(\kappa_1,...,\kappa_D)\,,\\
    L &= YD^{[L]}Y^{-1} \,,\quad\,\,\,\, D^{[L]}=\diag(\ell_1,...,\ell_D)\,,
\end{align*}
so that writing $Z=X^{-1}Y$ we have
\begin{align*}
    Ve^{K\otimes (a^\dagger)^2}e^{L\otimes a^\dagger}\ket{0}
    &= VXe^{D^{[K]}\otimes (a^\dagger)^2}Ze^{D^{[L]}\otimes a^\dagger}Y^{-1}\ket{0}\,. 
\end{align*}
Now define  
\begin{align}
    Z_{ij} \coloneqq \braket{\kappa_i|Z|\ell_j}\,,\quad W^{ij}\coloneqq \braket{\ell_i|Y^{-1}VX|\kappa_j}
\end{align}
and let $\mathsf{V,K,L}\in M_{D^2}(\C)$ be given by
\begin{equation}
    \begin{aligned}
        \mathsf{K} &= K\otimes \openone_D\,,\quad \mathsf{L} = \openone_D\otimes L\,,\\
        \mathsf{V} &\coloneqq (X\otimes Y)\sum_{ijkl=1}^D Z_{ij}W^{jk}\ketbra{ij}{kl}(X^{-1}\otimes Y^{-1}) \,.
    \end{aligned}
\end{equation}
Then the resulting BMPS is given by
\begin{align}
    \ket{\Psi_N} &= \Tr_{D^2}\left[\ordprod_{j=1}^N\mathsf{V}e^{\mathsf{K}\otimes (a_j^\dagger)^2}e^{\mathsf{L}\otimes a_j^\dagger}
    \right]\ket{0}^{\otimes N}\,,
\end{align}
which has commuting generator $[\mathsf{K},\mathsf{L}] = 0$ with bond dimension $D^2$. 

In what follows, without loss of generality we will focus on the commuting-generator GA-BMPS as a variational ansatz without losing expressivity from the non-commuting $K,L$ family. We will discuss some other alternative forms of the GA-BMPS family in Sec.~\ref{subsec: alternative-forms} including the case when $K,L$ are non-diagonalizable.

\subsection{Explicit computation}

Next, we show that we can calculate the transfer matrices of the GA-BMPS family explicitly, which in turn allows us to calculate physically relevant quantities. 

We first separate the BMPS transfer-matrix calculation from the
effect of the passive unitary. Since $\mathsf{U}$ is unitary, the
expectation value of a local observable $O_j$ is
\begin{align}
    \braket{O_j}_{\psi_N}
    &=
    \frac{
        \braket{
            \Psi_N|
            \mathsf{U}^\dagger O_j\mathsf{U}
            |\Psi_N
        }
    }{
        \braket{\Psi_N|\Psi_N}
    }\,,
    \label{eq: local-expectation-value}
\end{align}
where $\ket{\Psi_N}$ is the BMPS defined in
Eq.~\eqref{eq: exponential-ansatz-BMPS}. Let $\mathcal{U}\in \mathrm{U}(N)$ be the single-particle unitary associated with the passive unitary $\mathsf{U}$, i.e., 
\begin{align}
    \mathsf{U}^\dagger a_j\mathsf{U}
    =
    \sum_{k=1}^N \mathcal{U}_{jk}a_k\,.
    \label{eq: passive-observable-transformation}
\end{align}
It follows that if $O_j$ is a polynomial of fixed degree in $a_j,a_j^\dagger$, then $\mathsf{U}$ maps it to a polynomial of the same degree but potentially all modes $k$. Each resulting monomial can nonetheless be evaluated by standard transfer-matrix contraction, hence polynomial observables remain tractable. The situation is slightly different for Gaussian operators: for displacement operators,
\begin{align}
    \mathsf{U}^\dagger D_j(\alpha)\mathsf{U}
    =
    \bigotimes_{k=1}^N
    D_k\left(\bar{\mathcal{U}}_{jk}\alpha\right),
    \label{eq: passive-displacement-transformation}
\end{align}
while a single-mode squeezing operator becomes
\begin{align}
    \mathsf{U}^\dagger S_j(\zeta)\mathsf{U}
    =
    e^{\frac{1}{2}
        (
        \bar\zeta\, b_j^2
        -
        \zeta (b_j^\dagger)^2
        )}\,,
    \quad
    b_j
    \coloneqq
    \sum_{k=1}^N \mathcal{U}_{jk}a_k\,.
\end{align}
This shows that certain non-polynomial observables, notably displacement operators, remain tractable after conjugation by the passive unitary $\mathsf{U}$. In the remainder of this subsection we derive the local transfer matrices of the BMPS part by setting $\mathsf{U}=\openone$: the passive unitary can subsequently be included whenever $\mathsf{U}^\dagger O\mathsf{U}$ allows for tractable transfer-matrix contractions.

Since $\mathsf{U}$ is unitary, the norm of the GA-BMPS is equal to that
of its BMPS part:
\begin{align}
    \braket{\psi_N|\psi_N}
    =
    \braket{\Psi_N|\Psi_N}
    =
    \Tr\rr{
        (\bar B\otimes B)E^N
    },
    \label{eq: norm-convergence}
\end{align}
where the BMPS transfer matrix is given by
\begin{align}
    E
    =
    \sum_{n=0}^\infty
    \bar{A^n}\otimes A^n
    \label{eq: transfer-matrix-general}
\end{align}
The series converges absolutely when
$\rho(K)<1/2$, as shown in Appendix~\ref{appendix: ansatz-absolute-convergence}.

We would like to do better by computing $E$ explicitly as a function of $V,K,L$. We now show that for the commuting-generator BMPS family, there is a simple closed-form expression in terms of these matrices.
\begin{proposition}[Transfer matrix]
    \label{prop: transfer-matrix}
    Consider the family of commuting-generator BMPS with $[K,L]=0$ and $\rho(K)<\frac{1}{2}$. Then the transfer matrix $E$ is given by
    \begin{equation}
        \begin{aligned}
            E &=  \bar V \otimes{V}\mathcal{E}(\bar{L},\bar{K},K,L) \,,\\
            \mathcal{E}&\equiv \mathcal{E}(\bar{L},\bar{K},K,L)=\Delta^{\frac{1}{2}}
                e^{
                    \Delta\rr{
                        \bar K\otimes L^2
                        +\bar L\otimes L
                        +\bar L^2\otimes K
                    }
                }\,,\\
            \Delta
            &=
            \rr{
                \openone\otimes\openone
                -4\bar K\otimes K
            }^{-1}\,.
        \end{aligned}
        \label{eq: transfer-matrix-bosons}
    \end{equation}
    Here $\Delta^{1/2}$ denotes the principal matrix square root.
\end{proposition}
\noindent The proof involves straightforward but somewhat tedious algebraic manipulation involving the CCR algebra (Appendix~\ref{appendix: derivation-transfer}).

Next, the expectation value of local observables in standard MPS theory can be computed using the $O$-transfer matrix that we require to be convergent:
\begin{align}
    E_O
    =
    \sum_{i,j=0}^\infty
    \braket{i|O|j}
    \bar{A^i}\otimes A^j
\end{align}
provided that the defining series converges absolutely. In particular, we would like to be able to compute $E_O$ when $O=\mathsf{poly}(a,a^\dagger)$ or simple functions of $a,a^\dagger$, such as the displacement  $D(\alpha)$ or squeezing $S(\zeta)$. 

First, we show that we can compute exactly the expectation values of any arbitrary anti-normal-ordered monomials in $a,a^\dagger$:
\begin{proposition}[$O$-transfer matrix for $a^m(a^\dagger)^n$] 
    \label{prop: differentiate-E}
    Consider a ``sourced'' transfer operator
    \begin{align}
        E(s,t)
        \coloneqq
        (\bar V\otimes V)
        \mathcal{E}
        (
            \bar{L}+s\openone_D,
            \bar{K},
            K,
            L+t\openone_{D}
        )\,,
        \label{eq: source-dependent-transfer-matrix}
    \end{align}
    where $\mathcal{E}$ is given in Eq.~\eqref{eq: transfer-matrix-bosons}. 
    Then the $O$-transfer matrix $E_O =  \bar V \otimes{V}\mathcal{E}_O$ for anti-normal ordered operator $O=a^m(a^\dagger)^n$ is
    \begin{align}
        {E}_{a^m(a^\dagger)^n} 
        &= 
        \partial_s^m\partial_t^n E(s,t)\Bigr|_{s=t=0}\,.
        \label{eq: derivatives-transfer-matrix}
    \end{align}
\end{proposition}
\noindent The proof is by direct computation (Appendix~\ref{appendix: derivation-transfer}). This result extends to all $\mathsf{poly}(a,a^\dagger)$ using linearity and the CCR algebra. 

\begin{example}
Let $O = a, a^\dagger$. Then 
\begin{align*}
    E_a 
    &= \partial_{s}E(s,t)\bigr|_{s,t=0} =
    (\bar V\otimes V)
    \Delta
    \rr{
        \openone\otimes L
        +2\bar L\otimes K
    }
    \mathcal{E}\,,\\
    E_{a^\dagger}
    &= \partial_tE(s,t)\bigr|_{s,t=0} = 
    (\bar V\otimes V)
    \Delta
    \rr{
        \bar L\otimes\openone
        +2\bar K\otimes L
    }
    \mathcal{E}\,.
\end{align*}
For number operator $O=a^\dagger a$ we can proceed similarly. Using Proposition~\ref{prop: differentiate-E}, CCR algebra and linearity of the transfer matrix $E_{O_1+\alpha O_2} =E_{O_1} + \alpha E_{O_2}$, we have
    \begin{align*}
        E_n = E_{aa^\dagger} - E_\openone 
        &=
        (\partial_s\partial_t-1)E(s,t)
        \bigr|_{s=t=0}\,,
    \end{align*}
    where $\mathcal{E}$ is given in Eq.~\eqref{eq: transfer-matrix-bosons}. 
    Similarly, we have $n^2 = a^2(a^\dagger)^2 - 3aa^\dagger + \openone$ and hence
    \begin{align*}
        E_{n^2} = (\partial_s^2\partial_t^2
        -3\partial_s\partial_t+1)E(s,t)\bigr|_{s,t=0}\,.
    \end{align*}
    These calculations are sufficient to give the expectation value of local terms in the Bose-Hubbard model \cite{lewenstein2012ultracold}. 
\exampleqed
\end{example}

The transfer matrix can also be evaluated for certain non-polynomial operators, including single-mode Gaussian unitaries.
\begin{proposition}
    \label{prop: transfer-matrix-exponential-operators}
    Let $O=U_{\mathsf{G}}$ where $U_{\mathsf{G}}$ is a single-mode Gaussian unitary in Eq.~\eqref{eq: single-mode-gaussian-unitary}. Then writing
    $\tau=e^{i\phi}\tanh r$ and $c=\cosh r$, we have
    \begin{align}
        E_{U_{\mathrm G}}
        &=
        e^{i\chi}\sqrt{c}\,
        e^{\frac{1}{2}\rr{|\alpha|^2+\bar\tau\alpha^2}}\tilde{E}
    \end{align}
    where $\tilde{E}$ is the transfer matrix $E$ in Eq.~\eqref{eq: transfer-matrix-bosons} but with the following substitution
    \begin{align*}
        \bar{L}\otimes \openone &\mapsto (\bar{L}-(\bar\alpha+\bar\tau\alpha)\openone )\otimes \openone,\\
        \bar{K}\otimes \openone &\mapsto \rr{\bar{K}+\frac{1}{2}\bar\tau\openone} \otimes \openone\,,\\
        \openone \otimes K&\mapsto   
                \openone \otimes c^2\rr{e^{-2i\theta}K-\frac{1}{2}\tau\openone}
            \,,\\
        \openone\otimes L&\mapsto \openone\otimes ( c e^{-i\theta}L+\alpha\openone)\,.
    \end{align*}
\end{proposition}
\noindent The proof is given in Appendix~\ref{appendix: derivation-transfer}. Operators involving exponentials of $a,a^\dagger$ arise naturally in physical models such as the bosonized Schwinger model \cite{ohata2023montecarlo}.

\subsection{Alternative forms}
\label{subsec: alternative-forms}

There is an alternative representation of the BMPS ansatz in Def.~\ref{def: CV-MPS}, which motivates the following definition. Note that we are ignoring the Gaussian passive unitary $\mathsf{U}$ since it is not relevant for the discussion below.
\begin{definition}[MPS of photon-added Gaussian states]
    \label{def: photon-addition-MPS}
    Consider a family of photon-added Gaussian states
    \begin{align}
        \ket{I}
        \equiv
        \ket{q,\kappa_\lambda,\ell_\lambda},
        \qquad
        q=0,1,\ldots,\nu_\lambda,
    \end{align}
    where $\lambda=1,\ldots,m$ labels distinct pairs
    $(\kappa_\lambda,\ell_\lambda)$ satisfying
    $|\kappa_\lambda|<1/2$. We say that a bosonic many-body state is an
    MPS of photon-added Gaussian states if it takes the form
    \begin{align}
        \ket{\Psi_N^{\mathsf{PAG}}}
        =
        \sum_{I_1,\ldots,I_N=1}^{\mathsf d}
        \Tr\rr{
            B\mathsf{A}^{I_1}\cdots\mathsf{A}^{I_N}
        }
        \ket{I_1\cdots I_N},
        \label{eq: photon-added-MPS}
    \end{align}
    where $I=(q,\lambda)$, $B,\mathsf{A}^I\in M_D(\C)$, and
    the effective physical dimension is $\mathsf d \coloneqq \sum_{\lambda=1}^m(\nu_\lambda+1)$.
\end{definition}
\noindent Clearly, this definition also allows us to perform explicit computations: for example, the transfer matrix involves finite sums
\begin{align}
    E = \sum_{I,J=1}^{\mathsf{d}} \braket{I|J}\bar{\mathsf{A}^{I}}\otimes \mathsf{A}^J
\end{align}
where the cross terms arise due to non-orthogonality of $\ket{I}$. Expectation value of local observables involving $a,a^\dagger$ can also be computed efficiently due to properties of photon-added Gaussian states. 

The following proposition shows that the two definitions define the same set of bosonic quantum many-body states, hence we can use either formulation interchangeably depending on the problem at hand (Appendix~\ref{appendix: set-inclusion}).
\begin{proposition}
    \label{prop: CV-MPS-equivalence}
    Let $\ket{\Psi_N^{\exp}(D)}$ denote the exponential family of BMPS in Def.~\ref{def: CV-MPS} with $\mathsf{U}=\openone$ and $\ket{\Psi_{N}^{\mathsf{PAG}}(\mathsf{d},D)}$ the MPS of photon-added Gaussian states with effective physical dimension $\mathsf{d}$. 
    \begin{enumerate}[leftmargin=*,label=\emph{(\roman*)}]
        \item Every $\ket{\Psi_N^{\exp}(D)}$ can be expressed as $\ket{\Psi_N^{\mathsf{PAG}}(\mathsf{d},D)}$ with $\mathsf{d} \leq 2D^2-D$. 
        
        \item Conversely, every $\ket{\Psi_N^{\mathsf{PAG}}(\mathsf{d},D)}$ can be expressed as $\ket{\Psi_N^{\exp}(D\mathsf{d})}$ with commuting generators.
    \end{enumerate}
    Therefore, by taking the union over all finite bond dimensions and accounting for the passive unitary $\mathsf{U}$, the ansätze in
    Def.~\ref{def: CV-MPS} and \ref{def: photon-addition-MPS} generate the same class of bosonic many-body states.
\end{proposition}

At this point, it is natural to ask whether annihilation operators can
also be included in the exponential ansatz. To isolate this question,
consider the $K=0$ case
\begin{align}
    \ket{\tilde{\Psi}_N}
    \coloneqq
    \Tr_D\left[
        B\ordprod_{j=1}^N
        V e^{R\otimes a_j}e^{L\otimes a_j^\dagger}
    \right]\ket{0}^{\otimes N}.
\end{align}
If $[R,L]=0$, the factor involving $R$ can immediately be absorbed
into $V$ using the BCH formula. It turns out that even without the commutativity assumption, we can also absorb $R$ into the redefinition of $V$. Indeed,
\begin{align}
    e^{R\otimes a}e^{L\otimes a^\dagger}\ket{0}
    &=
    \sum_{m,n=0}^\infty
    \frac{R^mL^n}{m!n!}
    \otimes
    a^m(a^\dagger)^n\ket{0}.
\end{align}
Since
\begin{align}
    a^m(a^\dagger)^n\ket{0}
    &=
    \begin{cases}
        \displaystyle
        \frac{n!}{(n-m)!}
        (a^\dagger)^{n-m}\ket{0},
        & n\geq m,\\[6pt]
        0,
        & n<m,
    \end{cases}
\end{align}
setting $\ell=n-m$ gives
\begin{align}
    e^{R\otimes a}e^{L\otimes a^\dagger}\ket{0}
    &=
    \sum_{m,\ell=0}^\infty
    \frac{R^mL^{m+\ell}}{m!\ell!}
    \otimes(a^\dagger)^\ell\ket{0}
    \notag\\
    &=
    W e^{L\otimes a^\dagger}\ket{0},
    \label{eq: annihilation-absorption}
\end{align}
where
\begin{align}
    W
    \coloneqq
    \sum_{m=0}^\infty
    \frac{R^mL^m}{m!}.
\end{align}
The series defining $W$ converges absolutely for arbitrary $R,L\in M_D(\C)$ and can be evaluated through a finite-dimensional matrix exponential after vectorization. Substitution into the many-body ansatz yields
\begin{align}
    \ket{\tilde{\Psi}_N}
    =
    \Tr_D\left[
        B\ordprod_{j=1}^N
        (VW)e^{L\otimes a_j^\dagger}
    \right]\ket{0}^{\otimes N}.
\end{align}
Thus, for $K=0$, the factor $e^{R\otimes a_j}$ does not enlarge the
ansatz and can be absorbed through the redefinition $V\mapsto VW$.
When $[R,L]=0$, one has $W=e^{RL}$.

\subsection{Gauge freedom}

Understanding the gauge freedom of finite-dimensional MPS is important
for many reasons, including the classification of phases of matter in
one-dimensional systems
\cite{Schuch2011ClassifyingPhases,Chen2011CompleteClassification,
Pollmann2012SymmetryProtection}. It is also useful in variational
calculations, where a suitable gauge choice removes redundant parameters
and can improve numerical conditioning.

As in the finite-dimensional MPS setting, the matrices defining a
GA-BMPS do not uniquely specify the physical state. For any invertible
$X\in\mathrm{GL}(D,\C)$, the simultaneous transformations
\begin{equation}
    \begin{aligned}
        K&\mapsto XKX^{-1},
        &\qquad
        L&\mapsto XLX^{-1},\\
        V&\mapsto XVX^{-1},
        &
        B&\mapsto XBX^{-1}\,,
    \end{aligned}
\end{equation}
leaves the BMPS invariant. This virtual similarity transformation is independent of the dimension of the local physical Hilbert space. In the alternative representation of Def.~\ref{def: photon-addition-MPS}, the gauge freedom takes the standard MPS form
\begin{align}
    \mathsf{A}^I\mapsto X\mathsf{A}^I X^{-1},
    \qquad
    B\mapsto XBX^{-1}.
\end{align}
The nonorthogonality of the local states $\ket I$ modifies the transfer
matrix through their Gram matrix but does not affect this virtual gauge
transformation.

This gauge freedom is also useful in the variational calculations of Sec.~\ref{sec: variational}. In the commuting-generator subfamily, if
$K$ and $L$ are diagonalizable, then $[K,L]=0$ implies that they can be
simultaneously diagonalized. This simplifies the finite-dimensional
matrix functions entering the transfer matrix. If the joint eigenvalue
pairs $(\kappa_i,\ell_i)$ are distinct, the remaining gauge freedom
consists of permutations and invertible diagonal similarities. We can
order the joint eigenvalue pairs and impose suitable normalization
conditions on $V$ to fix this residual freedom in the generic case.

\subsection{Extensions}
\label{subsec: extensions}

The tensor-network structure of the BMPS allows several natural extensions. We briefly mention two of them.

First, the ansatz can accommodate multiple species of bosons per site.
Let $\{a_{\nu,j}^\dagger\}$ be bosonic creation operators for mode
$\nu=1,\ldots,M$ at site $j$, satisfying the CCR
\begin{align}
    [a_{\mu,j},a_{\nu,k}^\dagger]
    =
    \delta_{\mu\nu}\delta_{jk}\openone\,.
\end{align}
A simple multimode generalization of the BMPS is
\begin{align}
    &\ket{\Psi_N^M}
    \notag\\
    &\coloneqq
    \Tr_D\left[
        B\ordprod_{j=1}^N
        V_j\prod_{\nu=1}^M
        e^{K_{\nu,j}\otimes(a_{\nu,j}^\dagger)^2}
        e^{L_{\nu,j}\otimes a_{\nu,j}^\dagger}
    \right]\ket{0}^{\otimes N}\,,
    \label{eq: multispecies-BMPS}
\end{align}
where $B,V_j,K_{\nu,j},L_{\nu,j}\in M_D(\C)$, and $\ket{0}$
denotes the local $M$-mode vacuum, so that
$a_{\nu,j}\ket{0}=0$ for all $\nu,j$. The ordered products follow the
convention introduced in Eq.~\eqref{eq: exponential-ansatz}. As before,
the conditions $\rho(K_{\nu,j})<1/2$ are sufficient for
normalizability. For practical calculations, we restrict to
$[K_{\nu,j},L_{\nu,j}]=0$ for every $\nu,j$, which allows the transfer
operators to be computed directly. For uniform tensors, $V_j=V$, $K_{\nu,j}=K_\nu$, and $L_{\nu,j}=L_\nu$, the transfer matrix becomes
\begin{align}
    E
    =
    (\bar V\otimes V)
    \mathcal{E}_1\cdots\mathcal{E}_M\,,
\end{align}
where $\mathcal{E}_\nu$ is defined as in
Eq.~\eqref{eq: transfer-matrix-bosons}, with $K,L$ replaced by
$K_\nu,L_\nu$. More general on-site tensors may also contain cross-mode
quadratic terms proportional to $K_{\mu\nu,j}\otimes a_{\mu,j}^\dagger a_{\nu,j}^\dagger$ which give more complicated transfer operators: we restrict here to the product form in Eq.~\eqref{eq: multispecies-BMPS}, for which the transfer matrix factorizes over the modes. The Gaussian augmentation
is generalized by taking the passive unitary $\mathsf{U}$ to act on all
$MN$ modes.

Another natural extension is to mixed species. For instance, consider a
chain in which the odd sites are qudits and the even sites are bosonic
modes. Assuming that $N$ is even, a spin-boson MPS can be written as
\begin{align}
    &\ket{\Psi_N^{\mathsf{sb}}}
    \notag\\
    &\coloneqq 
    \sum_{\bm n}
    \Tr_D\left[
        B\ordprod_{r=1}^{N/2}
        \mathcal A^{n_r}
        e^{K\otimes(a_{2r}^\dagger)^2}
        e^{L\otimes a_{2r}^\dagger}
    \right]
    \ket{\bm n}_{\mathsf s}
    \ket{\bm 0}_{\mathsf b}\,,
\end{align}
where $\bm n=(n_1,\ldots,n_{N/2})\in\{0,\ldots,d-1\}^{N/2}$, and we use the shorthand
\begin{align}
    \ket{\bm n}_{\mathsf s}
    \ket{\bm 0}_{\mathsf b}
    \coloneqq
    \bigotimes_{r=1}^{N/2}
    \left(
        \ket{n_r}_{2r-1}\otimes\ket{0}_{2r}
    \right)\,.
\end{align}
Here $\mathcal A^n\in M_D(\C)$ is the MPS tensor associated with the
qudit basis state $\ket n$, while $K,L\in M_D(\C)$ specify the bosonic
tensor. A separate propagator matrix $V$ is unnecessary here because it
can be absorbed into the matrices $\mathcal A^n$. The transfer matrix
of one two-site cell factorizes as
\begin{align}
    E
    &=
    E_{\mathsf s}\mathcal{E}\,,
    \quad 
    E_{\mathsf s}
    =
    \sum_{n=0}^{d-1}
    \bar{\mathcal{A}^n}\otimes\mathcal A^n\,,
\end{align}
where $\mathcal{E}$ is defined in Eq.~\eqref{eq: transfer-matrix-bosons}. Alternatively, one may encode each qudit in the first $d$ Fock levels
of an auxiliary bosonic mode using Eq.~\eqref{eq: MPS-embedding}, which 
produces a two-site-periodic BMPS of bond dimension $Dd$. On the
encoded-qudit sites, $K=0$ and the matrices $V$ and $L$ are chosen as
in Eq.~\eqref{eq: MPS-embedding}, while the matrices on the bosonic
sites are extended trivially to the additional $d$-dimensional
auxiliary space.

\section{Parent Hamiltonian}
\label{sec: parent-hamiltonian}

In the standard finite-dimensional MPS framework, it is possible to construct parent Hamiltonians, for which an MPS is (one of) its exact ground states, that is local and frustration-free simply from the knowledge of the local tensors of the MPS. Since the bosonic MPS inherits most of the technology from the MPS, it is possible to construct the parent Hamiltonian the same way. Here for bosonic systems we seek a construction in terms of the ladder operators $a,a^\dagger$ using polynomials and, where necessary, exponentials of these operators. 

In this section we are interested in constructing the parent Hamiltonian for the BMPS $\ket{\Psi_N}$ without the passive unitary $\mathsf{U}$ (cf. Eq.~\eqref{eq: exponential-ansatz-BMPS}). If $H$ is a parent Hamiltonian for $\ket{\Psi_N}$, then $H_{\mathsf{U}} = \mathsf{U}H\mathsf{U}^\dagger$ is a parent Hamiltonian for the GA-BMPS family $\ket{\psi_N}=\mathsf{U}\ket{\Psi_N}$. The Hamiltonian $H_{\mathsf{U}}$ need not be local for a general passive unitary: strict locality is preserved, for example, if $\mathsf{U}$ is a finite-depth geometrically local Gaussian circuit.

\subsection{Parent Hamiltonian for MPS}

We first review the standard MPS construction
\cite{Fannes1992FinitelyCorrelated,PerezGarcia2007MPSRepresentations}.
For an MPS tensor $A=\{A^i\}_{i=1}^d$, define the $l$-site support space
\begin{align}
    G_l(A)
    \coloneqq
    \Span_\C
    \left\{
        \ket{\psi_l^{ab}(A)}
        :
        a,b=1,\ldots,D
    \right\},
\end{align}
where, for the matrix units
$E_{ab}\equiv\ketbra{a}{b}\in M_D(\C)$,
\begin{align}
    \ket{\psi_l^{ab}(A)}
    \coloneqq
    \sum_{i_1,\ldots,i_l}
    \Tr\left(
        E_{ab}A^{i_1}\cdots A^{i_l}
    \right)
    \ket{i_1\cdots i_l}.
\end{align}
By construction, $\dim G_l(A)\leq D^2$. 

Suppose that $l_0$ is an injectivity length, so that
\begin{align}
    \Span_\C
    \{
        A^{i_1}\cdots A^{i_{l_0}}
    \}
    =
    M_D(\C).
\end{align}
It follows that $\dim G_l(A)=D^2$ for every $l\geq l_0$. To obtain a
nonzero local parent term, one chooses $l\geq l_0$ such that $d^l>D^2$, 
so that $G_l(A)$ is a proper subspace of the $l$-site physical Hilbert
space. If $d^{l_0}=D^2$, then $G_{l_0}(A)$ fills the entire physical
space at the injectivity length, and one may take $l=l_0+1$.

\begin{definition}[Parent Hamiltonian for MPS]
    Let $\ket{\psi_N(A)}$ be an MPS and choose an interaction length
    $l$ such that $G_l(A)$ is a proper subspace of the $l$-site
    physical Hilbert space. Let $P_l$ be the orthogonal projector onto
    $G_l(A)^\perp$. For a periodic chain, the corresponding standard
    parent Hamiltonian is
    \begin{align}
        H
        \coloneqq
        \sum_{j=1}^{N}h_j,
        \qquad
        h_j
        \coloneqq
        \tau_j(P_l),
    \end{align}
    where $\tau_j(P_l)$ acts on sites
    $j,\ldots,j+l-1$, with the site labels understood modulo $N$.
\end{definition}

\noindent
Since every local reduced state of $\ket{\psi_N(A)}$ is supported on
$G_l(A)$,
\begin{align}
    h_j\ket{\psi_N(A)}=0
\end{align}
for every $j$. Thus $H\geq0$ is frustration-free and
$H\ket{\psi_N(A)}=0$.

If $A$ is injective on the present lattice and $d^2>D^2$, one may take
$l=2$. If injectivity is obtained only after blocking, the corresponding
two-site construction acts on two blocked sites. For a block-injective tensor, one may likewise take $l=2$ if each block is already injective without further blocking and $G_2(A)$ is a proper subspace. For sufficiently large periodic chains, the standard parent Hamiltonian of an injective MPS has a unique ground state and a spectral gap bounded below uniformly in $N$. If the tensor has $g$ distinct injective blocks, the ground-state degeneracy is $g$ independent of $N$. Open chains may have additional boundary degeneracy.

In principle, the projector construction remains valid when the local Hilbert space is infinite-dimensional. For bosonic systems, however, the projector onto
$G_l(A)^\perp$ need not have a useful expression in terms of $a,a^\dagger$. We therefore construct a positive local term with the same kernel directly in terms of $a,a^\dagger$.

\subsection{The abstract bosonic construction}
\label{subsec: abstract-parent}

Given a Fock space $\mathcal{F}$, consider an $n$-dimensional subspace
\begin{align}
    \mathcal{V}
    =
    \Span_\C
    \left\{
        \ket{e_j}:j=1,\ldots,n
    \right\}
\end{align}
where $\ket{e_j}$ are linearly independent vectors, and let
\begin{align}
    \mathcal{W}
    =
    \Span_\C
    \left\{
        \ket{\psi_i}
        =
        \sum_{j=1}^n c_{ij}\ket{e_j}
        :
        i=1,\ldots,m
    \right\}
\end{align}
be a proper subspace of $\mathcal{V}$ where $\ket{\psi_i}$ are $m$ linearly independent vectors with $m<n$. The parent Hamiltonian construction for a state $\ket{\psi}\in \mathcal{W}$ asks for a positive operator $h$ satisfying
\begin{align}
    \ker h=\mathcal{W}\,.
\end{align}

The construction consists of two parts and does not depend on any tensor-network assumptions. First, suppose that there exists a positive operator $h_0$ such that
\begin{align}
    \ker h_0=\mathcal{V}\,,
    \label{eq: ambient-Hamiltonian}
\end{align}
where we require in addition that $h_0$ is expressible as functions of creation and annihilation operators. As we will see later, the explicit construction of $h_0$ would make use of the structure of the state we are building the parent Hamiltonian for. 

Next, we want to single out the proper subspace $\mathcal{W}$ inside $\mathcal{V}$. For this, we define a matrix $C$ to be the coefficient matrix appearing in the definition of $\mathcal{W}$, i.e., 
\begin{align}
    C=[c_{ij}]\in M_{m,n}(\C)\,.
\end{align}
Since $\rank C=m$, we may choose a matrix $R\in M_{n-m,n}(\C)$ of rank $n-m$, called the \textit{check matrix}, such that
\begin{align}
    RC^T=0\,.
    \label{eq: check-matrix}
\end{align}
It follows that $ \ker R=\im C^T$. Now suppose, in addition, that there exist operators $F_j$ satisfying
\begin{align}
    F_j\ket{e_k}
    =
    \delta_{jk}\ket{\Omega}\,,
    \qquad
    j,k=1,\ldots,n\,.
    \label{eq: coefficient-extraction}
\end{align}
where $\ket{\Omega}$ is a fixed nonzero vector. For $\alpha=1,\ldots,n-m$, define
\begin{align}
    O_\alpha
    &\coloneqq
    \sum_{j=1}^n
    R_{\alpha j}F_j\,,\quad 
    h_R
    \coloneqq
    \sum_{\alpha=1}^{n-m}
    O_\alpha^\dagger O_\alpha\,.
    \label{eq: check-Hamiltonian}
\end{align}
Then for any
\begin{align}
    \ket{\phi}
    =
    \sum_{j=1}^n x_j\ket{e_j}
    \in\mathcal{V}\,,
\end{align}
we have
\begin{align}
    O_\alpha\ket{\phi}
    =
    (Rx)_\alpha\ket{\Omega}
\end{align}
and consequently,
\begin{align}
    \ker h_R\cap\mathcal{V}
    &=
    \left\{
        \sum_{j=1}^n x_j\ket{e_j}
        :
        x\in\ker R
    \right\}
    =
    \mathcal{W}\,.
\end{align}
Since $h_0$ and $h_R$ are positive,
\begin{align}
    \ker\left(h_0+h_R\right)
    &=
    \ker h_0\cap\ker h_R
    =
    \mathcal{W}
\end{align}
and therefore we can take $h=h_0+h_R$ to be the parent Hamiltonian for $\ket{\psi}\in\mathcal{W}$ with ground space $\mathcal{W}$. 

We will see in what follows that the structure of BMPS determines the target space $\mathcal W$ and provides us with explicit constructions of $h_0$ and $F_j$. All expressions involving unbounded operators are understood on a suitable common dense domain.

\subsection{Examples}

We first illustrate the construction for a single mode, as these examples
will provide us the single-mode operators used in the $N$-site BMPS
construction that follows. We omit the passive unitary $\mathsf{U}$ whose
effect has already been discussed.

\subsubsection{Single mode}
\label{subsec: single-modes}

After absorbing the boundary matrix into $V$, the single-mode state is
\begin{align}
    \ket{\psi}
    =
    \Tr_D\left[
        V e^{K\otimes(a^\dagger)^2}
        e^{L\otimes a^\dagger}
    \right]\ket{0}\,.
\end{align}

\begin{example}[Superposition of coherent states]
    \label{example: cats}

    Let $K=0$ and
    \begin{align}
        L
        =
        \operatorname{diag}(\ell_1,\ldots,\ell_D)\,,
        \qquad
        V\in M_D(\C)\,,
    \end{align}
    where the $\ell_j$ are distinct. Writing $c_j\coloneqq V_{jj}$, we
    obtain
    \begin{align}
        \ket{\psi}
        =
        \sum_{j=1}^D c_j\ket{\ell_j}\,,
        \qquad
        \ket{\ell_j}
        \coloneqq
        e^{\ell_j a^\dagger}\ket{0}\,,
    \end{align}
    where the coherent states are unnormalized. The ambient and target
    spaces are
    \begin{align*}
        \mathcal{V}
        &=
        \Span_\C
        \left\{
            \ket{\ell_j}:j=1,\ldots,D
        \right\}\,,
        &
        \mathcal{W}
        &=
        \Span_\C\{\ket{\psi}\}\,.
    \end{align*}

    We first construct $h_0$ whose kernel is the ambient space
    $\mathcal{V}$. Since $(a-\ell_j)\ket{\ell_j}=0$, we define
    \begin{align}
        Q(a)
        &\coloneqq
        \prod_{j=1}^D(a-\ell_j)\,,
        \quad 
        h_0
        \coloneqq
        Q(a)^\dagger Q(a)\,.
    \end{align}
    On the one hand, each factor $a-\ell_j$ has a one-dimensional kernel and
    $\dim\ker(AB)\leq\dim\ker A+\dim\ker B$, therefore we have $\dim\ker Q\leq D$. On the other hand, $\ker Q$ contains the $D$ linearly independent eigenvectors $\ket{\ell_1},\ldots,\ket{\ell_D}$ of $a$. Therefore
    \begin{align}
        \ker h_0
        =
        \ker Q
        =
        \mathcal{V}\,.
    \end{align}

    Having fixed the ambient kernel, we next construct $h_R$ to select the
    target subspace $\mathcal{W}\subset\mathcal{V}$. Since our goal is to construct $F_j$ in Eq.~\eqref{eq: coefficient-extraction}, a natural construction is to use Lagrange interpolation polynomials
    \begin{align}
        \mathsf{P}_j(z)
        \coloneqq
        \prod_{k\neq j}
        \frac{z-\ell_k}{\ell_j-\ell_k}
    \end{align}
    that satisfy
    \begin{align}
        \mathsf{P}_j(a)\ket{\ell_k}
        =
        \delta_{jk}\ket{\ell_j}\,.
    \end{align}
    We may therefore take
    \begin{align}
        F_j
        \coloneqq
        e^{-\ell_j a^\dagger}\mathsf{P}_j(a)\,,
        \qquad
        F_j\ket{\ell_k}
        =
        \delta_{jk}\ket{0}\,.
    \end{align}
    The coefficient matrix is
    \begin{align}
        C
        =
        \begin{bmatrix}
            c_1&c_2&\cdots&c_D
        \end{bmatrix}\,.
    \end{align}
    Assuming $\ket{\psi}\neq0$, we may relabel the coherent states so
    that $c_D\neq0$. One possible choice of $R$ is
    \begin{align}
        R
        =
        \begin{bmatrix}
            c_D & 0   & \cdots & 0   & -c_1 \\
            0   & c_D & \cdots & 0   & -c_2 \\
            \vdots & \vdots & \ddots & \vdots & \vdots \\
            0   & 0   & \cdots & c_D & -c_{D-1}
        \end{bmatrix}\,,
    \end{align}
    which has rank $D-1$ and satisfies $RC^T=0$. It follows that
    \begin{align}
        h_R
        =
        \sum_{\alpha=1}^{D-1}
        \left(
            c_DF_\alpha-c_\alpha F_D
        \right)^\dagger
        \left(
            c_DF_\alpha-c_\alpha F_D
        \right)\,.
    \end{align}
    Thus $h=h_0+h_R$ is positive and satisfies $\ker h=\mathcal{W}$.
    \exampleqed
\end{example}

\begin{example}[Non-diagonalizable $L$]
    \label{example: non-diagonalizable-cats}

    Let $K=0$ and
    \begin{equation}
    \begin{aligned}
        J_2(\ell_1)
        &\coloneqq
        \begin{bmatrix}
            \ell_1&1\\
            0&\ell_1
        \end{bmatrix}\,,
        \qquad
        L
        =
        J_2(\ell_1)\oplus J_1(\ell_2)\,,
        \\
        V
        &=
        \begin{bmatrix}
            c_1/2&0&0\\
            c_2&c_1/2&0\\
            0&0&c_3
        \end{bmatrix}\,,
    \end{aligned}
    \end{equation}
    where $\ell_1\neq\ell_2$ and $c_j\neq0$. We then obtain
    \begin{align}
        \ket{\psi}
        =
        c_1\ket{\ell_1}
        +
        c_2a^\dagger\ket{\ell_1}
        +
        c_3\ket{\ell_2}\,.
    \end{align}
    Defining
    \begin{align*}
        \ket{e_1}
        &\coloneqq
        \ket{\ell_1}\,,
        \quad
        \ket{e_2}
        \coloneqq
        a^\dagger\ket{\ell_1}\,,
        \quad 
        \ket{e_3}
        \coloneqq
        \ket{\ell_2}\,,
    \end{align*}
    the ambient and target spaces are
    \begin{align*}
        \mathcal{V}
        &=
        \Span_\C
        \left\{
            \ket{e_1},\ket{e_2},\ket{e_3}
        \right\}\,,
        &
        \mathcal{W}
        &=
        \Span_\C\{\ket{\psi}\}\,.
    \end{align*}

    We first construct $h_0$ with $\ker h_0=\mathcal{V}$. Since
    \begin{align}
        (a-\ell_1)\ket{\ell_1}
        &=
        0\,,
        \quad 
        (a-\ell_1)^2a^\dagger\ket{\ell_1}
        =
        0\,,
    \end{align}
    we take
    \begin{align}
        Q(a)
        &\coloneqq
        (a-\ell_1)^2(a-\ell_2)\,,
        \quad
        h_0
        \coloneqq
        Q(a)^\dagger Q(a)\,.
    \end{align}
    We have $\dim\ker Q\leq3$ and $Q$ annihilates $\ket{e_1},\ket{e_2},\ket{e_3}$ which are linearly independent. Therefore $\ker h_0 = \ker Q = \mathcal{V}$.

    Next, we construct $h_R$ to single out the target subspace $\mathcal{W}\subset\mathcal{V}$. We first distinguish the generalized eigenspace associated with $\ell_1$ from the eigenspace associated with $\ell_2$. Let $\Delta \ell=\ell_1-\ell_2$ and define
    \begin{align}
        \mathsf{P}_1(z)
        &\coloneqq
        1-\frac{(z-\ell_1)^2}{(\Delta\ell)^2}\,,
        &
        \mathsf{P}_2(z)
        &\coloneqq
        \frac{(z-\ell_1)^2}{(\Delta\ell)^2}\,.
    \end{align}
    These polynomials satisfy
    \begin{equation*}
        \begin{aligned}
            \mathsf{P}_1(a)\ket{e_1}
            &=
            \ket{e_1}\,,
            &
            \mathsf{P}_1(a)\ket{e_2}
            &=
            \ket{e_2}\,,
            &
            \mathsf{P}_1(a)\ket{e_3}
            &=
            0\,,
            \\
            \mathsf{P}_2(a)\ket{e_1}
            &=
            0\,,
            &
            \mathsf{P}_2(a)\ket{e_2}
            &=
            0\,,
            &
            \mathsf{P}_2(a)\ket{e_3}
            &=
            \ket{e_3}
        \end{aligned}
    \end{equation*}
    and hence the $\mathsf{P}_j(a)\ket{e_k}=0$ whenever their coherent amplitudes differ. 

    It remains to distinguish $\ket{e_1}$ from the photon-added state
    $\ket{e_2}$ within the same generalized eigenspace with coherent amplitude $\ell_1$. On their span, the operator
    \begin{align}
        \mathsf{N}_1
        \coloneqq
        a^\dagger(a-\ell_1)
    \end{align}
    satisfies
    \begin{align}
        \mathsf{N}_1\ket{e_1}
        &=
        0\,,
        &
        \mathsf{N}_1\ket{e_2}
        &=
        \ket{e_2}\,.
    \end{align}
    We may therefore take
    \begin{equation}
    \begin{aligned}
        F_{1,0}
        &\coloneqq
        e^{-\ell_1a^\dagger}
        (1-\mathsf{N}_1)\mathsf{P}_1(a)\,,
        \\
        F_{1,1}
        &\coloneqq
        a e^{-\ell_1a^\dagger}
        \mathsf{N}_1\mathsf{P}_1(a)\,,
        \\
        F_{2,0}
        &\coloneqq
        e^{-\ell_2a^\dagger}\mathsf{P}_2(a)\,.
    \end{aligned}
    \end{equation}
    These operators obey
    \begin{equation}
    \begin{aligned}
        F_{1,0}\ket{e_k}
        &=
        \delta_{1k}\ket{0}\,,
        \\
        F_{1,1}\ket{e_k}
        &=
        \delta_{2k}\ket{0}\,,
        \\
        F_{2,0}\ket{e_k}
        &=
        \delta_{3k}\ket{0}\,.
    \end{aligned}
    \end{equation}

    The coefficient matrix is
    \begin{align}
        C
        =
        \begin{bmatrix}
            c_1&c_2&c_3
        \end{bmatrix}\,.
    \end{align}
    One possible choice of $R$ is
    \begin{align}
        R
        =
        \begin{bmatrix}
            c_3&0&-c_1\\
            0&c_3&-c_2
        \end{bmatrix}\,,
    \end{align}
    which has rank $2$ and satisfies $RC^T=0$. It follows that
    \begin{align}
        h_R
    	&=
    	\left(
    	c_3F_{1,0}-c_1F_{2,0}
    	\right)^\dagger
    	\left(
    	c_3F_{1,0}-c_1F_{2,0}
    	\right)
    	\notag\\
        &+
    	\left(
    	c_3F_{1,1}-c_2F_{2,0}
    	\right)^\dagger
    	\left(
    	c_3F_{1,1}-c_2F_{2,0}
    	\right)\,.
    \end{align}
    Thus $h=h_0+h_R$ is positive and satisfies
    $\ker h=\mathcal{W}$.

    \exampleqed
\end{example}

When $K$ is nonzero, the ambient space may contain states with different squeezing parameters that makes the parent Hamiltonian slightly more complicated. That said, one can still construct the parent Hamiltonian systematically and we provide a general construction for the operators $Q$ and $F_j$ in Appendix~\ref{appendix: parent-details}.

\subsubsection{Multimode BMPS}
\label{subsec: MPS-parents}

The extension from single-mode to full BMPS over $N$ sites is straightforward with some minor modifications. For this, we first set up the notation to make the prescription manifest. 

Write the BMPS in the photon-added form
\begin{align}
	\ket{\psi_N}
	=
	\sum_{I_1,\ldots,I_N=1}^{\mathsf d}
	\Tr\rr{
		B\mathsf{A}^{I_1}\cdots\mathsf{A}^{I_N}
	}
	\ket{I_1\cdots I_N}\,,
\end{align}
where the one-site states $\{\ket I\}_{I=1}^{\mathsf d}$ are linearly independent. Their $l$-site products define an ambient space
\begin{align}
	\mathcal{V}_l
	\coloneqq
	\Span_\C
	\left\{
		\ket{I_1\cdots I_l}:I_j=1,\ldots,\mathsf d
	\right\}\cong \mathcal{V}_1^{\otimes l}
\end{align}
where $n\coloneqq \dim \mathcal{V}_l = \mathsf{d}^l$. For each $X\in M_D(\C)$, define 
\begin{align}
	\ket{\Psi_l(X)}
	\coloneqq
	\sum_{I_1,\ldots,I_l=1}^{\mathsf d}
	\Tr\rr{
		X\mathsf{A}^{I_1}\cdots\mathsf{A}^{I_l}
	}
	\ket{I_1\cdots I_l}\,.
\end{align}
The corresponding $l$-site local MPS subspace is
\begin{align}
	G_l
	\coloneqq
	\left\{
		\ket{\Psi_l(X)}
		:
		X\in M_D(\C)
	\right\}
	\subseteq
	\mathcal{V}_l\,.
\end{align}
Choose $l$ such that $m\coloneqq\dim G_l<n$ so that $G_l^\perp\cap \mathcal{V}_l\neq 0$. 

Since the states
$\{\ket{\Psi_l(E_{ab})}\}_{a,b=1}^{D}$ span $G_l$, where $E_{ab}$ are the matrix units in $M_D(\C)$, we may choose matrices $X_1,\ldots,X_m$ such that
\begin{align}
	\ket{\Psi_\mu}
	\coloneqq
	\ket{\Psi_l(X_\mu)}\,,
	\qquad
	\mu=1,\ldots,m,
\end{align}
form a basis of $G_l$. Writing $\boldsymbol I=(I_1,\ldots,I_l)$, define the coefficient matrix
\begin{align}
	C_{\mu,\boldsymbol I}
	\coloneqq
	\Tr\rr{
		X_\mu
		\mathsf{A}^{I_1}\cdots\mathsf{A}^{I_l}
	}\,.
\end{align}
Then $C\in M_{m,n}(\C)$ has rank $m$.

Since the one-site states $\{\ket I\}_{I=1}^{\mathsf d}$ form a finite linearly independent family of photon-added squeezed coherent states, Appendix~\ref{appendix: parent-details} provides one-site
operators $Q$ and $F_I$ satisfying
\begin{equation}
    \begin{aligned}
        \ker Q
    	&=\mathcal{V}_1\,,\quad 
        F_I\ket J
    	=
    	\delta_{IJ}\ket{0}
    \end{aligned}
\end{equation}
as discussed in the single-mode setting. On $l$ sites, set
\begin{align}
	h_0
	&=
	\sum_{r=1}^{l}
	Q_r^\dagger Q_r\,,
	\quad 
	F_{\boldsymbol I}^{[l]}
	=
	F_{I_1}\otimes\cdots\otimes F_{I_l}\,,
\end{align}
where $Q_r$ acts on the $r$th site. It follows that
\begin{align}
	\ker h_0
	&=
	\mathcal{V}_l\,,
	\quad 
	F_{\boldsymbol I}^{[l]}
	\ket{\boldsymbol J}
	=
	\delta_{\boldsymbol I,\boldsymbol J}
	\ket{0}^{\otimes l}\,.
\end{align}
Choose a matrix $R\in M_{n-m,n}(\C)$ of rank $n-m$ such that the check matrix condition $RC^T=0$ (cf. Eq.~\eqref{eq: check-matrix}), and set
\begin{align}
	O_\alpha
	&=
	\sum_{\boldsymbol I}
	R_{\alpha,\boldsymbol I}
	F_{\boldsymbol I}^{[l]}\,,
	\quad 
	h_R
	=
	\sum_{\alpha=1}^{n-m}
	O_\alpha^\dagger O_\alpha\,.
\end{align}
The construction of Sec.~\ref{subsec: abstract-parent} then gives
\begin{align}
	\ker(h_0+h_R)
	=
	G_l\,.
\end{align}
Translating this $l$-site term along the chain gives a
frustration-free parent Hamiltonian for the BMPS. The single-mode construction in Appendix~\ref{appendix: parent-details} guarantees the existence of the
required operators $Q$ and $F_I$ for the general BMPS local basis, and
therefore completes the parent-Hamiltonian construction for the BMPS.

We illustrate the multimode construction with two examples.

\begin{example}[GHZ-like bosonic MPS]
    \label{example: N-cats}

    Consider the multimode analogue of Example~\ref{example: cats}. Let $K=0$ and take
    \begin{align}
        L
        &=
        \diag(\ell_1,\ldots,\ell_D)\,,\quad 
        V
        =
        B
        =
        \openone\,,
    \end{align}
    where the $\ell_j$ are distinct. The corresponding MPS matrices are
    \begin{align}
        \mathsf{A}^j
        =
        \ketbra{j}{j}\,,
        \qquad
        j=1,\ldots,D,
    \end{align}
    and hence
    \begin{align}
        \ket{\Psi_N}
        =
        \sum_{j=1}^{D}
        \ket{\ell_j}^{\otimes N}\,.
    \end{align}
    Here
    $\ket{\ell_j}=e^{\ell_j a^\dagger}\ket{0}$ is unnormalized. The two-site ambient space is
    \begin{align}
        \mathcal{V}_2
        =
        \Span_\C
        \left\{
            \ket{\ell_j\ell_k}
            :
            j,k=1,\ldots,D
        \right\},
    \end{align}
    with $\dim\mathcal{V}_2=D^2$. Taking
    $X_i=E_{ii}$ gives
    \begin{align}
        \ket{\Psi_2(X_i)}
        =
        \ket{\ell_i\ell_i}\,,
    \end{align}
    so that
    \begin{align}
        G_2
        =
        \Span_\C
        \left\{
            \ket{\ell_i\ell_i}
            :
            i=1,\ldots,D
        \right\},
    \end{align}
    and $\dim G_2=D$. Thus $G_2$ is a proper subspace of
    $\mathcal{V}_2$.

    For the ambient space, we define
    \begin{align}
        Q_i
        &\coloneqq
        \prod_{j=1}^{D}
        (a_i-\ell_j)\,,
        \quad 
        h_0
        \coloneqq
        \sum_{i=1}^{2}
        Q_i^\dagger Q_i\,.
    \end{align}
    Then $\ker h_0 = \mathcal{V}_2$. As in Example~\ref{example: cats}, let
    \begin{align}
        F_j
        &\coloneqq
        e^{-\ell_j a^\dagger}\mathsf{P}_j(a)\,,
        &
        \mathsf{P}_j(z)
        &\coloneqq
        \prod_{k\neq j}
        \frac{z-\ell_k}{\ell_j-\ell_k}\,.
    \end{align}
    These operators satisfy $F_j\ket{\ell_k}=\delta_{jk}\ket{0}$, thus we take
    \begin{align}
        F_{jk}^{[2]}
        \coloneqq
        F_j\otimes F_k\,.
    \end{align}
    With columns indexed by $(j,k)$, the coefficient matrix $C$ for the
    basis $\{\ket{\Psi_2(X_i)}\}_{i=1}^{D}$ is
    \begin{align}
        C_{i,(j,k)}
        =
        \delta_{ij}\delta_{ik}\,.
    \end{align}
    Choose the rows of $R$ to be indexed by the ordered pairs
    $(r,s)$ with $r\neq s$, and set
    \begin{align}
        R_{(r,s),(j,k)}
        =
        \delta_{rj}\delta_{sk}\,.
    \end{align}
    Then $RC^T=0$ with $\rank R=D^2-D$. The corresponding check operators are
    \begin{align}
        O_{(r,s)}=F_{rs}^{[2]}\,, \quad h_R
        =
        \sum_{\substack{r,s=1\\r\neq s}}^{D}
        \left(F_{rs}^{[2]}\right)^\dagger
        F_{rs}^{[2]}\,,
    \end{align}
    noting that the sum only runs over $r\neq s$. This gives $\ker(h_0+h_R)=G_2$ as required. 

    \exampleqed
\end{example}

\begin{example}[Injective MPS, $D=2$]
    \label{example: generic-D2-diagonal-L}

    Let $K=0$ and take
    \begin{align}
        L
        &=
        \diag(\ell_1,\ell_2)\,,
        &
        V
        &=
        \begin{bmatrix}
            v_{11} & v_{12}\\
            v_{21} & v_{22}
        \end{bmatrix},
    \end{align}
    where $\ell_1,\ell_2$ are distinct, $V$ is invertible, and
    $v_{rs}\neq0$. The one-site states are
    $\ket{\ell_1},\ket{\ell_2}$, and the corresponding MPS matrices are
    \begin{align}
        \mathsf{A}^r
        =
        V E_{rr}\,,
        \qquad
        r=1,2\,,
    \end{align}
    where $E_{rs}=\ketbra{r}{s}$ are matrix units. Their two-site products satisfy
    \begin{align}
        \mathsf{A}^r\mathsf{A}^s
        =
        v_{rs}V E_{rs}\,.
    \end{align}
    Since $V$ is invertible and $v_{rs}\neq0$, these four products span
    $M_2(\C)$. The tensor is therefore injective after blocking two
    sites. In particular,
    \begin{align}
        \dim G_2
        =
        4
        =
        \dim\mathcal{V}_2\,,
    \end{align}
    and hence $G_2=\mathcal{V}_2$. Consequently, there is no nontrivial
    two-site check term within the ambient space and we need to consider blocking $l=3$ sites.

    For $l=3$ the ambient space is
    \begin{align}
        \mathcal{V}_3
        =
        \Span_\C
        \left\{
            \ket{\ell_r\ell_s\ell_t}
            :
            r,s,t=1,2
        \right\},
    \end{align}
    with $\dim\mathcal{V}_3=8$, whereas $\dim G_3=4$: indeed, we have
    \begin{align}
        \mathsf{A}^r\mathsf{A}^s\mathsf{A}^t
        =
        v_{rs}v_{st}V E_{rt}\,.
    \end{align}
    Choose $X_{rt}=E_{tr}V^{-1}$ so that the corresponding states spanning $G_3$ are given by
    \begin{align}
        \ket{\Psi_{rt}}
        \coloneqq
        \ket{\Psi_3(X_{rt})}
        =
        \sum_{s=1}^{2}
        v_{rs}v_{st}
        \ket{\ell_r\ell_s\ell_t}\,,
    \end{align}
    and form a basis of $G_3$.

    With columns indexed by $(p,s,u)$, the coefficient matrix $C$ is given by
    \begin{align}
        C_{(r,t),(p,s,u)}
        =
        \delta_{rp}\delta_{tu}
        v_{rs}v_{st}\,.
    \end{align}
    A rank-four check matrix $R$ satisfying $RC^T=0$ is
    \begin{align}
        R_{(r,t),(p,s,u)}
        =
        \delta_{rp}\delta_{tu}
        \left(
            \delta_{s1}v_{r2}v_{2t}
            -
            \delta_{s2}v_{r1}v_{1t}
        \right).
    \end{align}

    Similar to Example~\ref{example: cats}, define
    \begin{align}
        Q_i
        &\coloneqq
        (a_i-\ell_1)(a_i-\ell_2)\,,
        \quad 
        h_0
        \coloneqq
        \sum_{i=1}^{3}Q_i^\dagger Q_i\,.
    \end{align}
    which gives $\ker h_0=\mathcal{V}_3$. Using the single-site operators $ F_j=e^{-\ell_j a^\dagger}\mathsf{P}_j(a)$, set
    \begin{align}
        F_{rst}^{[3]}
        \coloneqq
        F_r\otimes F_s\otimes F_t\,.
    \end{align}
    The four check operators obtained from the rows of $R$ are
    \begin{align}
        O_{rt}
        =
        v_{r2}v_{2t}F_{r1t}^{[3]}
        -
        v_{r1}v_{1t}F_{r2t}^{[3]}\,,
        \quad
        r,t=1,2\,.
    \end{align}
    Thus the check Hamiltonian is given by
    \begin{align}
        h_R
        =
        \sum_{r,t=1}^{2}
        O_{rt}^\dagger O_{rt}
    \end{align}
    and we have $\ker(h_0+h_R)=G_3$.
    \exampleqed
\end{example}

In Sec.~\ref{sec: variational}, we use this parent Hamiltonian to test the GA-BMPS ansatz variationally.

\section{Variational tests}
\label{sec: variational}

In this section we provide some numerical demonstration that the GA-BMPS \eqref{eq: exponential-ansatz} in Def.~\ref{def: CV-MPS} can be used in practice for calculations without introducing a hard local Fock space cutoff\footnote{Early attempt at non-Gaussian truncation-free variational calculations was done in \cite{frenzel2013matrix} in the context of spin-boson models, which corresponds to the case where the matrices are chosen to be $K=0$ (no squeezing), $V_j=\openone,L_j\neq 0$ but the $L_j$'s are not uniform. Since dense matrices $L_j$ were used, the calculations in \cite{frenzel2013matrix} were restricted to small bond dimensions $D$.}. We do not attempt to optimize the variational method for speed or performance and focus primarily on the physical results that can be extracted from the ansatz. The results shown here should be viewed as the baseline that is expected to improve with better schemes. Below we present two small examples.

\subsection{Quartic interactions}

The $(1+1)$-dimensional real scalar field theory with quartic
interactions, conventionally denoted by $\phi^4_2$, has the formal
Hamiltonian
\begin{align}
    H_{\phi^4}
    \coloneqq
    \int\dd x\,
    \left[
        \frac{1}{2}\pi^2
        +
        \frac{1}{2}(\partial_x\phi)^2
        +
        \frac{1}{2}m^2\phi^2
        +
        g\,{:}\phi^4{:}
    \right],
\end{align}
where $g$ is the quartic coupling, $m>0$ is the bare mass
parameter, and
\begin{align}
    [\phi(x),\pi(y)]
    =
    i\delta(x-y)\,.
\end{align}
Here the quartic interaction is Wick ordered with respect to the
vacuum of the free theory of mass $m$. The $\phi^4_2$ model is one of the standard examples of an interacting relativistic quantum field theory that admits a rigorous construction \cite{GlimmJaffee1968phi4-1,GlimmJaffee1968phi4-2,GlimmJaffee1968phi4-3}. It has also been studied numerically using Hamiltonian truncation, lattice MPS, Monte Carlo methods, and tensor-network
renormalization \cite{Rychkov2015phi4,Vanhecke2019qft, Milsted2013criticalQFT,Bosetti2015montecarlophi4,Bronzin2015montecarlophi4,delcamp2020rgphi4} (see also \cite{tilloy2026reviewRCMPS} for a recent review). 

Our aim here is not to obtain a precise determination of the critical properties of $\phi^4_2$, but to use it as a test of the GA-BMPS ansatz for a Hamiltonian without $U(1)$ symmetry. We discretize space while retaining the full local Fock space and work directly in the thermodynamic limit. The dimensionless lattice Hamiltonian is
\begin{align}
    aH
    =
    \sum_j
    \left[
        \frac{1}{2}\pi_j^2
        +
        \frac{1}{2}(\phi_{j+1}-\phi_j)^2
        +
        \frac{1}{2}(am)^2\phi_j^2
        +
        a^2g\,{:}\phi_j^4{:}
    \right],
\end{align}
where the discretized field operators are
\begin{equation*}
    \begin{aligned}
        \phi_j
        &\coloneqq
        \phi(aj)
        =
        \frac{a_j+a_j^\dagger}{\sqrt{2}}\,,\quad 
        \pi_j
        \coloneqq
        a\pi(aj)
        =
        \frac{i(a_j^\dagger-a_j)}{\sqrt{2}}\,.
    \end{aligned}
\end{equation*}
These operators satisfy $[\phi_j,\pi_k]=i\delta_{jk}$ and the lattice Wick-ordering prescription is \cite{Rychkov2015phi4}
\begin{align}
    {:}\phi_j^4{:}=\phi_j^4-6G\phi_j^2+3G^2\,.
\end{align}
Here the lattice dispersion $\omega(q)=\sqrt{(am)^2+4\sin^2(q/2)}$ and $G$ is the coincident two-point function of the free massive lattice vacuum in the thermodynamic limit, namely
\begin{align}
    G
    \coloneqq
    \langle\phi_j^2\rangle_0
    =
    \int_{-\pi}^{\pi}
    \frac{\dd q}
    {4\pi\sqrt{(am)^2+4\sin^2(q/2)}}
\end{align}
The same lattice dispersion and coincident propagator appear in Hamiltonian lattice studies of $\phi^4_2$ \cite{Milsted2013criticalQFT}. 

\begin{figure*}[tp]
    \centering
    \includegraphics[scale=0.33]{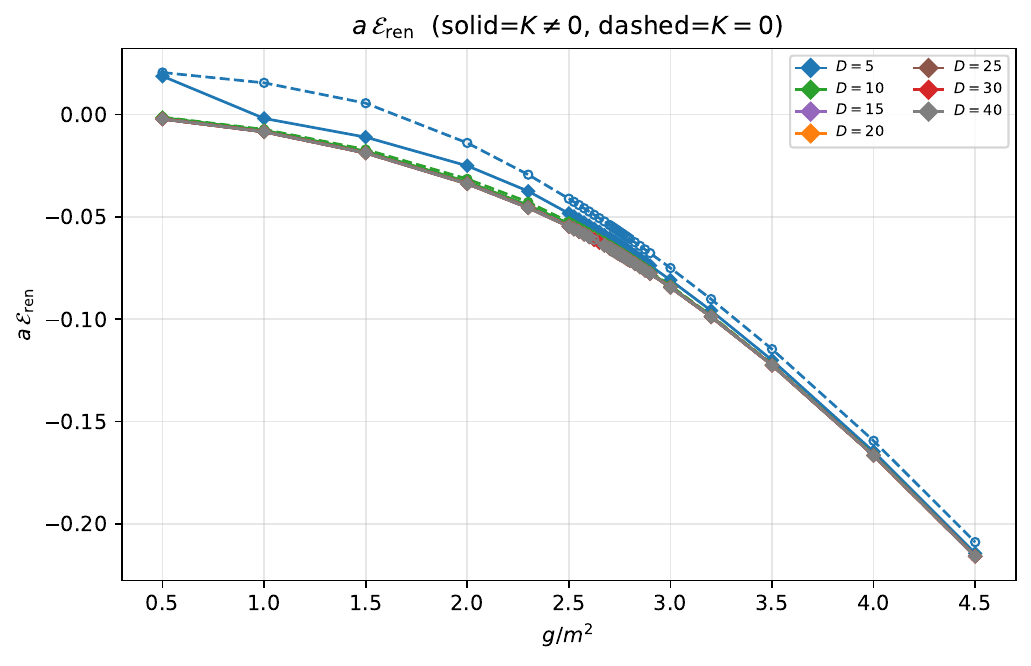}
    \includegraphics[scale=0.33]{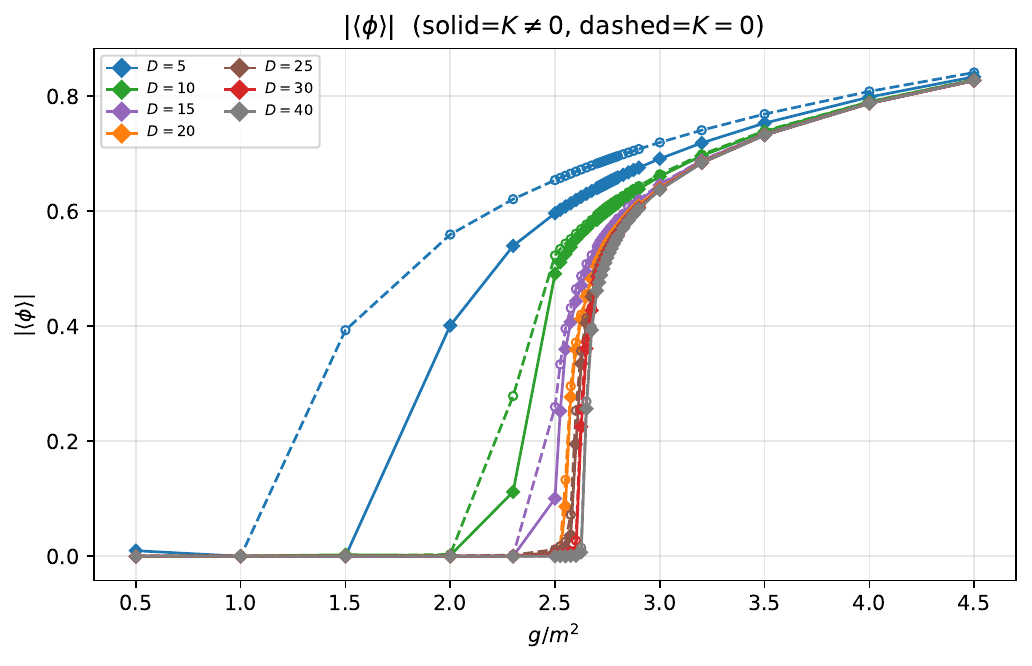}
    \includegraphics[scale=0.33]{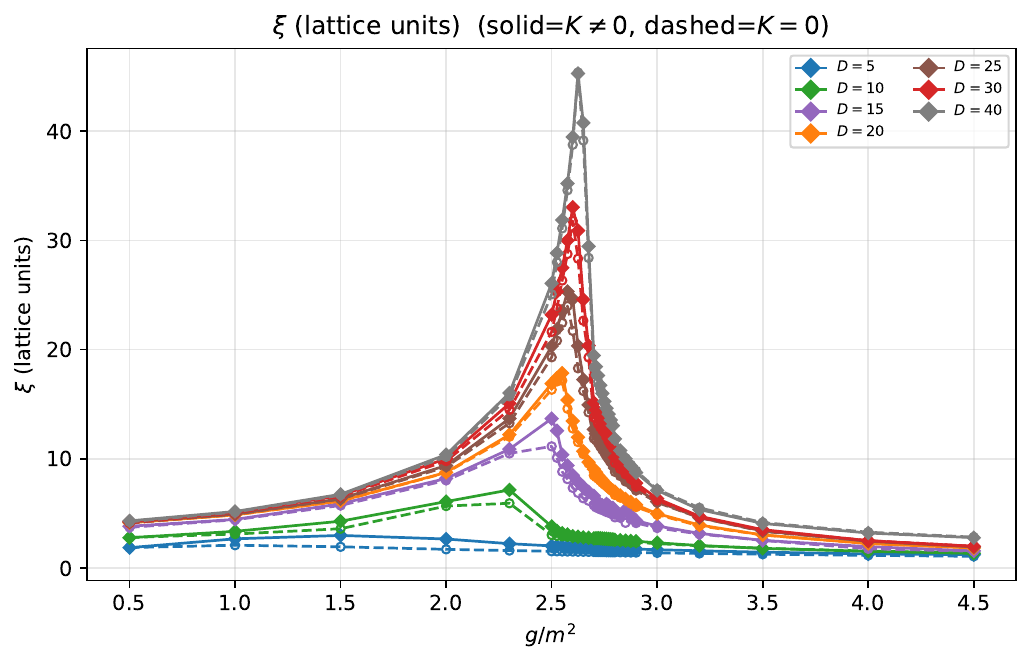}
    \caption{
    Thermodynamic-limit BMPS results for the lattice $\phi^4_2$ model at $a=0.2$ and $m=1$. Solid diamonds show the on-site-squeezed ansatz with $K\neq0$, while dashed open circles show the $K=0$ ansatz. \textit{Left:} renormalized energy $a\mathcal E_{\rm ren}$; \textit{Middle:} the magnitude of the order parameter; \textit{Right:}
    the correlation length in lattice units. For the parameters and bond dimensions shown, including on-site squeezing lowers the variational energy. The order parameter and correlation length display a finite-$D$ crossover consistent with the Ising transition.
    }
    \label{fig: phi4-plots}
\end{figure*}

For completeness, we also specify the energy density shown below. Writing
\begin{align}
    aH=\sum_j h_j\,,
\end{align}
the zero-point energy per site of the free lattice theory is
\begin{align}
    \epsilon_{0,\infty}
    =
    \bra{0_m}h_j\ket{0_m}\big|_{g=0}
    =
    \int_{-\pi}^{\pi}
    \frac{\dd q}{4\pi}\,\omega(q)\,.
\end{align}
The rescaled renormalized energy density is given by
\begin{align}
    a\mathcal E_{\mathrm{ren}}
    \coloneqq
    \frac{1}{a}
    \left(
        \langle h_j\rangle-\epsilon_{0,\infty}
    \right).
    \label{eq: renormalized-energy-density}
\end{align}
Thus the Wick-ordering prescription supplies the quadratic
counterterm in the Hamiltonian, whereas the subtraction of
$\epsilon_{0,\infty}$ fixes the free-vacuum energy to zero. No
continuum extrapolation is implied by the notation
$\mathcal E_{\mathrm{ren}}$. We use this lattice Hamiltonian as our
testbed. 

The numerical results are shown in Fig.~\ref{fig: phi4-plots}. We
restrict to the commuting, simultaneously diagonalizable subfamily and
use the gauge freedom to represent $K$ and $L$ in diagonal form, with
$V$ absorbing the corresponding similarity transformation. Bulk
expectation values are computed from the dominant left and right fixed
points of the transfer operator. The variational energy was minimized
in double precision using L-BFGS-B, with gradients obtained by automatic
differentiation through the fixed-point calculation. For each coupling, the $K=0$ optimization was performed both from seeded random initializations and, where available, from the optimized parameters at the preceding value of $g$ and we retain the lowest-energy candidate. The $K\neq0$ optimization was initialized from the selected $K=0$ state together with additional seeded squeezing initializations, and the lowest-energy candidate was retained. At every value of $D$ and $g$, the optimized ansatz with $K\neq0$ has lower variational energy than the corresponding $K=0$ ansatz.

\subsection{Ground state of parent Hamiltonian}

In Section~\ref{sec: parent-hamiltonian} we provided the parent Hamiltonian for the BMPS family. It is therefore expected that if one starts from such a parent Hamiltonian $H$, the BMPS should be able to variationally find the ground state of $H$. Furthermore, for injective BMPS the uniqueness of the ground state should allow us to recover the matrices that specify the BMPS. We demonstrate this using Example~\ref{example: generic-D2-diagonal-L}, which gives a generic three-site parent Hamiltonian with $D=2$ and no squeezing.

We considered the BMPS with parameters
\begin{align*}
    V
    &=
    \begin{pmatrix}
        0.696 & 0.287 \\
        -0.164+0.235i & 0.592
    \end{pmatrix}\,,\\
    K&=0\,, \qquad L=
    \diag(-1.000,1.300)
\end{align*}
and used the parent Hamiltonian constructed in Example~\ref{example: generic-D2-diagonal-L}. Here $E_\rho$ denotes its thermodynamic-limit energy density, evaluated from the dominant fixed points of the BMPS transfer map. We minimized $E_\rho$ using L-BFGS-B, with gradients obtained by automatic differentiation. For each bond dimension $D=1,2,3$, we retained the lowest-energy result among six seeded random initializations, none of which used the target parameters. We define $D^\ast$ as the smallest bond dimension for which
$E_\rho<10^{-7}$. As shown in Table~\ref{tab:parent-recovery}, this
occurs at $D^\ast=2$. Increasing the bond dimension to $D=3$ produces
only a numerically vanishing third Schmidt weight, while the two leading
weights remain essentially unchanged, consistent with a ground state of
bond dimension $D=2$.

\begin{table}[t]
    \centering
    \small
    \begin{tabular}{@{}ccc@{}}
        \toprule
        $D$ & $E_\rho$ & Schmidt weights $\{s_i\}$ \\
        \midrule
        1 & $3.9\times10^{-6}$
          & $\{1.000\}$ \\
        2 & $4.7\times10^{-9}$
          & $\{0.868,\,0.132\}$ \\
        3 & $3.6\times10^{-9}$
          & $\{0.870,\,0.130,\,4.0\times10^{-16}\}$ \\
        \bottomrule
    \end{tabular}
    \caption{Variational recovery of the injective MPS in Example~\ref{example: generic-D2-diagonal-L}. The Schmidt weights correspond to the eigenvalues of the reduced density matrix across half the chain.}
    \label{tab:parent-recovery}
\end{table}

As shown in Example~\ref{example: generic-D2-diagonal-L}, the BMPS
tensor becomes injective after blocking two sites, while its nontrivial
parent Hamiltonian is 3-local. Consequently, we expect that the corresponding infinite-volume parent Hamiltonian has a unique ground state. In this case, one should be able to recover the matrices $V,L$ accurately when the ground state energy is found. Indeed, at $D^\ast=2$ we show that the optimized state also recovers the target parameters after gauge fixing. After gauge fixing, the recovered parameters agree with the prescribed parameters with maximum entrywise errors
\begin{equation}
\begin{aligned}
    \max|V_{\mathrm{rec}}-V_{\mathrm{target}}|
    &=
    1.1\times10^{-7}\,,
    \\
    \max|\ell_{\mathrm{rec}}
             -\ell_{\mathrm{target}}|
    &=
    9.2\times10^{-10}\,.
\end{aligned}
\end{equation}

\section{Discussion and outlook}
\label{sec: discussions}

In this paper we propose and analyze the structure of a rich family of bosonic many-body states called Gaussian-augmented bosonic matrix-product states (GA-BMPS) that have the following features: (i) they include all pure Gaussian states and truncated (finite-dimensional) MPS; (ii) expectation values can be efficiently computed and hence, among other things, they can be used for variational calculations; (iii) they have exact parent Hamiltonians that are simple functions of the bosonic creation and annihilation operators. The GA-BMPS family contains important families of states such as the Gaussian states and the finite-dimensional MPS (by embedding them into Fock space). 

The structure of the ansatz opens up to several future directions. The first  is higher-dimensions: while this does not solve the hardness of PEPS contraction in general, such a bosonic PEPS ansatz may become relevant when Fock space truncation is preferably avoided. Second, just like the standard MPS, since the ansatz enables species mixtures of bosons, fermions and spins, it would be interesting to see if they can be applied for studying (non-)equilibrium properties of mixed-species problems. {Last but not least, it would be interesting to develop a more sophisticated algorithms that exploit the natural structure of the ansatz for both static and dynamical problems. While we have shown that for static problems the GA-BMPS family can be effective, we expect that much work needs to be done before it can be used for dynamical problems beyond standard approaches based on truncated MPS or Gaussian states.} We leave these open for future work.

\paragraph*{Note added.}
{While this manuscript was being completed, we became aware of ~\cite{brauer2026symmetry}, which constructs oscillator parent Hamiltonians for multimode GHZ-, cluster-, and W-type cat resource states using coherent-branch, correlation, and symmetry constraints. The constructions overlap for coherent-state cat subfamilies (e.g., their branch operator $a^2-\alpha^2$ corresponds to our $Q$), whereas our construction applies systematically to general BMPS whose local basis states include photon-added squeezed coherent states.} 

\section*{Acknowledgment}

The authors are grateful to Mari Carmen Bañuls, Marco Rigobello, Jutho Haegeman, Wei Tang, and Daniel Malz for insightful discussions.  E.T. acknowledges support from the Alexander von Humboldt Foundation.  J.I.C acknowledges funding by THEQUCO as part of the Munich Quantum Valley,
which is supported by the Bavarian state government with funds from the Hightech Agenda Bayern Plus. This work has been partially supported by
Klaus Tschira Foundation.

\section*{AI disclosure}

All numerical calculations were performed using AI-generated Python
code developed with GPT-5.5 and Claude Code Opus-4.8 under iterative human supervision. The implementations were subsequently audited using GPT-5.6 Sol (Codex). All scientific results are checked by the authors and the authors take responsibility for them.

\bibliography{ref}

\appendix

\section{From the symplectic Bloch-Messiah decomposition to Gaussian unitaries}
\label{appendix: Bloch-Messiah}

Here we review the decomposition of Gaussian unitaries used in Sec.~\ref{sec: GA-BMPS} from the symplectic Bloch-Messiah decomposition. 

We use the quadrature ordering introduced in Sec.~\ref{sec: GA-BMPS},
\begin{align}
    \mathbf R
    \coloneqq
    \begin{bmatrix}
        \bm{x}\\
        \bm{p}
    \end{bmatrix}
    =
    (x_1,\ldots,x_N,p_1,\ldots,p_N)^T\,,
\end{align}
for which
\begin{align}
    [R_j,R_k]
    &=
    i\Omega_{jk}\openone\,,
    &
    \Omega
    &=
    \begin{bmatrix}
        0&\openone_N\\
        -\openone_N&0
    \end{bmatrix}\,.
\end{align}
An $N$-mode Gaussian unitary acts affinely on the quadratures
\cite{Weedbrook2012GaussianQI},
\begin{align}
    U_{\mathsf G}^\dagger
    \mathbf R
    U_{\mathsf G}
    =
    S\mathbf R+\mathbf d\,,
    \label{eq: affine-symplectic-action}
\end{align}
where $\mathbf d\in\R^{2N}$ and $S\in\mathrm{Sp}(2N,\R)$ satisfies $S\Omega S^T=\Omega$.

We first recall the symplectic factorization of $S$. Let
\begin{align}
    P
    &\coloneqq
    (SS^T)^{1/2}\,,
    \quad
    O
    \coloneqq
    P^{-1}S\,.
    \label{eq: symplectic-polar-decomposition}
\end{align}
Polar decomposition $S=PO$ gives a positive symplectic matrix $P>0$
and an orthogonal symplectic matrix $O$, i.e., 
\begin{align}
    P\Omega P
    &=
    \Omega\,,
    \quad
    O\Omega O^T
    =
    \Omega\,,
    \quad
    OO^T
    =
    \openone_{2N}\,.
    \label{eq: symplectic-polar-properties}
\end{align}
Since $P=P^T>0$, the first relation is equivalently
\begin{align}
    P\Omega
    =
    \Omega P^{-1}\,.
    \label{eq: positive-symplectic-intertwining}
\end{align}
If $Pv=\lambda v$ with $\lambda>0$ for some nonzero vector $v\in \R^{2N}$, then
\begin{align}
    P(\Omega v)
    = \Omega P^{-1}v
    = \lambda^{-1}\Omega v\,.
    \label{eq: reciprocal-symplectic-eigenvalues}
\end{align}
Therefore, the eigenvalues of $P$ occur in reciprocal pairs (if $\lambda$ is a nonzero eigenvalue then so is $\lambda^{-1}$). Since $P$ is real and symmetric, these pairs can be arranged into an
orthonormal symplectic eigenbasis: that is, there exists an orthogonal symplectic matrix
\begin{align}
    O_{\mathsf{U}}
    \in
    \mathrm{Sp}(2N,\R)\cap\mathrm O(2N)\,,
\end{align}
such that
\begin{align}
    P
    &=
    O_{\mathsf{U}}
    \Sigma
    O_{\mathsf{U}}^T\,,
    \label{eq: positive-symplectic-diagonalization}
    \\
    \Sigma
    &=
    \begin{bmatrix}
        e^{-\bm r}&0\\
        0&e^{\bm r}
    \end{bmatrix}\,,
    &
    \bm r
    &=
    \diag(r_1,\ldots,r_N)\,,
    \label{eq: diagonal-symplectic-squeezing}
\end{align}
where $r_j \geq 0$. Since $O_{\mathsf W}\coloneqq O_{\mathsf{U}}^T O$ 
also belongs to
$\mathrm{Sp}(2N,\R)\cap\mathrm O(2N)$, it follows that $S$ admits a symplectic singular value decomposition (the so-called Bloch-Messiah decomposition) \cite{Braunstein2005bloch-messiah}:
\begin{align}
    S
    =
    O_{\mathsf{U}}
    \Sigma
    O_{\mathsf W}\,.
    \label{eq: symplectic-Bloch-Messiah}
\end{align}
The quantities $e^{\pm r_j}$ (squeezing) are the symplectic singular values of
$S$. Like the standard singular-value decompositions, the above decomposition is not unique. 

\begin{proposition}[Fock-space form of a Gaussian unitary]
\label{prop: Bloch-Messiah}
Let $U_{\mathsf G}$ be an $N$-mode Gaussian unitary satisfying
Eq.~\eqref{eq: affine-symplectic-action}. Then, up to a global phase,
\begin{align}
    U_{\mathsf G}
    =
    \mathsf{U}
    \left[
        \prod_{j=1}^N
        D_j(\alpha_j)S_j(\zeta_j)
    \right]
    \mathsf W\,,
    \label{eq: Gaussian-unitary-Bloch-Messiah}
\end{align}
where $\mathsf{U}$ and $\mathsf W$ are passive linear-optical
unitaries, while $D_j(\alpha_j)$ and $S_j(\zeta_j)$ are displacement
and squeezing operators acting on mode $j$, respectively.
\end{proposition}

\begin{proof}
We lift the three factors in Eq.~\eqref{eq: symplectic-Bloch-Messiah} to Fock space as Gaussian unitaries.

First consider an arbitrary orthogonal symplectic matrix
\begin{align}
    O
    \in
    \mathrm{Sp}(2N,\R)\cap\mathrm O(2N)\,.
\end{align}
Orthogonality and symplecticity imply that $O\Omega=\Omega O$, hence $O$ can be expressed as blocks of $N\times N$ matrices
\begin{align}
    O
    =
    \begin{bmatrix}
        X&-Y\\
        Y&X
    \end{bmatrix}
    \label{eq: orthogonal-symplectic-block}
\end{align}
where the orthogonality condition $O^TO=\openone_{2N}$ implies that
\begin{align}
    X^TX+Y^TY
    &=
    \openone_N\,,
    &
    X^TY-Y^TX
    &=
    0\,.
    \label{eq: orthogonal-symplectic-block-relations}
\end{align}
If we define $\mathcal{U}\coloneqq X+iY$, we obtain
\begin{align}
    \mathcal{U}^\dagger \mathcal{U}
    &=
    (X^T-iY^T)(X+iY)
    =
    \openone_N\,,
\end{align}
hence $\mathcal{U}$ is unitary.

To identify the required transformation on the annihilation
operators, consider the action of $O$ on the quadratures $\mathbf{R}\equiv [\bm{x},\bm{p}]^T$:
\begin{align}
    U_{O}^\dagger
    \begin{bmatrix}
        \bm{x}\\
        \bm{p}
    \end{bmatrix}
    U_{O}
    &=
    \begin{bmatrix}
        X&-Y\\
        Y&X
    \end{bmatrix}
    \begin{bmatrix}
        \bm{x}\\
        \bm{p}
    \end{bmatrix}
    =
    \begin{bmatrix}
        X\bm{x}-Y\bm{p}\\
        Y\bm{x}+X\bm{p}
    \end{bmatrix}\,.
    \label{eq: passive-quadrature-action}
\end{align}
Since
\begin{align}
    \bm{a}
    \coloneqq
    (a_1,\ldots,a_N)^T
    =
    \frac{\bm{x}+i\bm{p}}{\sqrt2}\,,
\end{align}
we obtain
\begin{align}
    U_{O}^\dagger
    \bm{a}
    U_{O}
    &=
    \frac{
    X\bm{x}-Y\bm{p}
        +
        i(Y\bm{x}+X\bm{p})       
    }{\sqrt2}
    \notag\\
    &=
    (X+iY)
    \frac{\bm{x}+i\bm{p}}{\sqrt2}
    \notag\\
    &=
    \mathcal{U}\bm{a}\,.
    \label{eq: passive-mode-transformation}
\end{align}
Applying this to $O_{\mathsf{U}}$ and
$O_{\mathsf W}$, we can write the corresponding unitary $N\times N$ matrices 
\begin{align}
    \mathcal{U}_{\mathsf{U}}
    &=
    e^{i\mathsf{M}}\,,\quad
    \mathcal{U}_{\mathsf W}
    =
    e^{i\mathsf{N}}\,,
\end{align}
where $\mathsf M$ and $\mathsf N$ are Hermitian. Define
\begin{align}
    \mathsf{U}
    &\coloneqq
    \exp\rr{
        i\sum_{j,k=1}^N
        \mathsf M_{jk}
        a_j^\dagger a_k
    }\,,
    \label{eq: passive-unitary-U}
    \\
    \mathsf W
    &\coloneqq
    \exp\rr{
        i\sum_{j,k=1}^N
        \mathsf N_{jk}
        a_j^\dagger a_k
    }\,.
    \label{eq: passive-unitary-W}
\end{align}
For a Hermitian matrix $\mathsf{H}$, the quadratic operator
\begin{align}
    G_{\mathsf{H}}
    \coloneqq
    \sum_{j,k=1}^N
    \mathsf{H}_{jk}a_j^\dagger a_k
\end{align}
satisfies
\begin{align}
    [G_{\mathsf{H}},a_l]
    &=
    \sum_{j,k=1}^N
    \mathsf{H}_{jk}
    [a_j^\dagger a_k,a_l]
    =
    -\sum_{k=1}^N
    \mathsf{H}_{l k}a_k\,.
\end{align}
The BCH formula therefore gives
\begin{align}
    e^{-iG_{\mathsf{H}}}
    \bm{a}
    e^{iG_{\mathsf{H}}}
    =
    e^{i\mathsf{H}}\bm{a}\,.
\end{align}
It follows that
\begin{align}
    \mathsf{U}^\dagger a_j \mathsf{U}
    &=
    \sum_{k=1}^N [\mathcal{U}_{\mathsf{U}}]_{jk}a_k\,,
    &
    \mathsf W^\dagger a_j\mathsf W
    &=
    \sum_{k=1}^N [\mathcal{U}_{\mathsf{W}}]_{jk}a_k\,,
\end{align}
hence the Gaussian unitaries $\mathsf{U}$ and $\mathsf W$ implement orthogonal symplectic transformation $O_{\mathsf{U}}$ and $O_{\mathsf W}$, respectively. 
These unitaries preserve the total number operator:
\begin{align*}
    \mathsf{U}^\dagger
    n_{\mathsf{tot}}
    \mathsf{U}
    &=
    n_{\mathsf{tot}}\,,\quad 
    n_{\mathsf{tot}}
    \coloneqq
    \sum_{j=1}^N a_j^\dagger a_j\,,
\end{align*}
and similarly for $\mathsf W$, so both are \textit{passive linear-optical unitaries}.

The diagonal symplectic matrix $\Sigma$ is implemented by
\begin{align}
    \mathsf S(\bm r)
    &\coloneqq
    \prod_{j=1}^N S_j(r_j)\,,
    \quad 
    S_j(r_j)
    \coloneqq
    e^{\frac{1}{2}r_j
        \left(
            a_j^2-(a_j^\dagger)^2
        \right)}\,.
    \label{eq: product-squeezing}
\end{align}
The single-mode transformations are
\begin{equation}
    \begin{aligned}
        S_j(r_j)^\dagger x_kS_j(r_j)
        &=
        \begin{cases}
            e^{-r_j}x_j\,,&k=j\,,\\
            x_k\,,&k\neq j\,,
        \end{cases}
        \\
        S_j(r_j)^\dagger p_kS_j(r_j)
        &=
        \begin{cases}
            e^{r_j}p_j\,,&k=j\,,\\
            p_k\,,&k\neq j\,.
        \end{cases}
    \end{aligned}
\end{equation}
Therefore,
\begin{align}
    \mathsf S(\bm r)^\dagger
    \begin{bmatrix}
        \bm{x}\\
        \bm{p}
    \end{bmatrix}
    \mathsf S(\bm r)
    &=
    \begin{bmatrix}
        e^{-\bm r}\bm{x}\\
        e^{\bm r}\bm{p}
    \end{bmatrix}
    =
    \Sigma\mathbf R\,.
    \label{eq: squeezing-symplectic-lift}
\end{align}

Finally, we need to account for the displacement vector $\mathbf d$. Define
\begin{align}
    D(\bm\alpha)
    &\coloneqq
    \prod_{j=1}^N D_j(\alpha_j)\,,
    \quad 
    D_j(\alpha_j)
    \coloneqq
    e^{\alpha_j a_j^\dagger
        -
        \bar\alpha_j a_j}\,.
\end{align}
Since $D(\bm\alpha)^\dagger a_jD(\bm\alpha)=a_j+\alpha_j$, we have
\begin{equation}
    \begin{aligned}
        D(\bm\alpha)^\dagger x_jD(\bm\alpha)
        &=
        x_j+\sqrt2\,\Re(\alpha_j)\,,
        \\
        D(\bm\alpha)^\dagger p_jD(\bm\alpha)
        &=
        p_j+\sqrt2\,\Im(\alpha_j)\,.
    \end{aligned}
\end{equation}
Thus
\begin{align}
    D(\bm\alpha)^\dagger
    \mathbf R
    D(\bm\alpha)
    =
    \mathbf R+\bm\delta\,,
\end{align}
where $\bm\delta = [ \bm\delta_{\bm{x}},\bm\delta_{\bm{p}}]^T$ and
\begin{align}
    \delta_{x,j}
    &=
    \sqrt2\,\Re(\alpha_j)\,,
    \quad
    \delta_{p,j}
    =
    \sqrt2\,\Im(\alpha_j)\,.
    \label{eq: displacement-quadrature-relation}
\end{align}

Now set
\begin{align}
    \widetilde U_{\mathsf G}
    \coloneqq
    \mathsf{U}
    D(\bm\alpha)
    \mathsf S(\bm r)
    \mathsf W\,.
    \label{eq: candidate-Bloch-Messiah}
\end{align}
Its Heisenberg action can be evaluated in the order in which the
operators occur:
\begin{align}
    \widetilde U_{\mathsf G}^\dagger
    \mathbf R
    \widetilde U_{\mathsf G}
    &=
    \mathsf W^\dagger
    \mathsf S(\bm r)^\dagger
    D(\bm\alpha)^\dagger
    \mathsf{U}^\dagger
    \mathbf R
    \mathsf{U}
    D(\bm\alpha)
    \mathsf S(\bm r)
    \mathsf W
    \notag\\
    &=
    \mathsf W^\dagger
    \mathsf S(\bm r)^\dagger
    D(\bm\alpha)^\dagger
    \left(
        O_{\mathsf{U}}\mathbf R
    \right)
    D(\bm\alpha)
    \mathsf S(\bm r)
    \mathsf W
    \notag\\
    &=
    \mathsf W^\dagger
    \mathsf S(\bm r)^\dagger
    O_{\mathsf{U}}
    \left(
        \mathbf R+\bm\delta
    \right)
    \mathsf S(\bm r)
    \mathsf W
    \notag\\
    &=
    \mathsf W^\dagger
    O_{\mathsf{U}}
    \left(
        \Sigma\mathbf R+\bm\delta
    \right)
    \mathsf W
    \notag\\
    &=
    O_{\mathsf{U}}
    \left(
        \Sigma O_{\mathsf W}\mathbf R+\bm\delta
    \right)
    \notag\\
    &=
    O_{\mathsf{U}}\Sigma O_{\mathsf W}\mathbf R
    +
    O_{\mathsf{U}}\bm\delta
    \notag\\
    &=
    S\mathbf R+O_{\mathsf{U}}\bm\delta\,.
    \label{eq: verify-Bloch-Messiah}
\end{align}
If we choose
\begin{align}
    \bm\delta
    =
    O_{\mathsf{U}}^T\mathbf d\,,
\end{align}
then $\tilde{U}_\mathsf{G}$ reproduces the affine action
in Eq.~\eqref{eq: affine-symplectic-action}. Consequently, the unitaries $U_{\mathsf G}$ and $\widetilde U_{\mathsf G}$ have the same action on all
$a_j$ and $a_j^\dagger$, which implies in particular that 
$U_{\mathsf G}\widetilde U_{\mathsf G}^\dagger$ commutes with every
canonical operator. The Fock representation of the canonical
commutation relations is irreducible, so
\begin{align}
    U_{\mathsf G}
    =
    e^{i\chi}
    \widetilde U_{\mathsf G}
\end{align}
for some $\chi\in\R$.

The decomposition above uses real squeezing parameters $r_j$. The
more general single-mode squeezing operator
\begin{align}
    S_j(\zeta_j)
    \coloneqq
    e^{\frac{1}{2}(\bar\zeta_j a_j^2-\zeta_j(a_j^\dagger)^2)}
\end{align}
with $\zeta_j=r_je^{i\phi_j}$ differs from $S_j(r_j)$ by passive
phase rotations. These rotations can be absorbed into
$\mathsf{U}$ and $\mathsf W$, giving
Eq.~\eqref{eq: Gaussian-unitary-Bloch-Messiah}.
\end{proof}

Since operators acting on distinct modes commute,
\begin{align}
    D(\bm\alpha)\mathsf S(\bm r)
    =
    \prod_{j=1}^N
    D_j(\alpha_j)S_j(r_j)\,,
\end{align}
which gives the product appearing in
Eq.~\eqref{eq: Gaussian-unitary-Bloch-Messiah}.

For the multimode vacuum, the rightmost passive unitary contributes
only a global phase. Indeed, preservation of the total number operator
implies
\begin{align}
    n_{\mathsf{tot}}
    \mathsf W\ket{0}^{\otimes N}
    &=
    \mathsf W
    n_{\mathsf{tot}}\ket{0}^{\otimes N}
    =
    0\,.
\end{align}
The zero-particle subspace is one-dimensional, and hence
\begin{align}
    \mathsf W\ket{0}^{\otimes N}
    =
    \ket{0}^{\otimes N}\,.
\end{align}
It follows that
\begin{align}
    U_{\mathsf G}\ket{0}^{\otimes N}
    =
    \mathsf{U}
    \bigotimes_{j=1}^N
    \left[
        D_j(\alpha_j)S_j(\zeta_j)\ket{0}_j
    \right]
\end{align}
up to a global phase, as used in
Eq.~\eqref{eq: bloch-messiah-decomposition-states}. The remaining
passive unitary $\mathsf{U}$ is generally nontrivial.

\section{Absolute convergence of the transfer matrix}
\label{appendix: ansatz-absolute-convergence}

\begin{lemma}
    Consider a bosonic MPS with tensor
    \begin{align}
        A^n
        =
        \sqrt{n!}\,V
        \sum_{m=0}^{\lfloor n/2\rfloor}
        \frac{K^mL^{n-2m}}{m!(n-2m)!}\,,
        \qquad
        n\in\mathbb N_0\,.
    \end{align}
    If $\rho(K)<1/2$, then the transfer-matrix series
    \begin{align}
        E
        \coloneqq
        \sum_{n=0}^{\infty}
        \bar{A^n}\otimes A^n
    \end{align}
    converges absolutely.
\end{lemma}

\begin{proof}
    Let $\lVert K\rVert_2=\sqrt{\lambda_{\max}(K^\dagger K)}$ denote the spectral norm of $K$. Choose $k$ such that
    \begin{align}
        \rho(K)<k<\frac{1}{2}\,.
    \end{align}
    Using Gelfand's formula for spectral radius
    \begin{align}
        \rho(K) = \lim_{n\to\infty}||K^n||^{1/n}
    \end{align}
    where $||\cdot||$ is any matrix norm, there exists $n_0\in \N$ such that 
    \begin{align}
        ||K^n||_2 \leq k^n\qquad \forall n\geq n_0\,.
    \end{align}
    Consequently, there exists some constant $\mathcal{C}_{K,k}>0$ such that 
    \begin{align}
        ||K^n||_2\leq \mathcal{C}_{K,k} k^n\qquad \forall n\in \mathbb{N}_0\,.
        \label{eq: bound-powers-K}
    \end{align}
    Writing
    \begin{align}
        A^n
        &=
        \sqrt{n!}\,VC_n\,,
        \quad
        C_n
        =
        \sum_{m=0}^{\lfloor n/2\rfloor}
        \frac{K^mL^{n-2m}}{m!(n-2m)!}\,,
    \end{align}
    and setting $l\coloneqq\lVert L\rVert_2$, sub-multiplicativity of the spectral norm gives
    \begin{align}
        \lVert {C}_n\rVert_2
        &\leq
        \sum_{m=0}^{\lfloor n/2\rfloor}
        \frac{
            \lVert K^m\rVert_2
            \lVert L^{n-2m}\rVert_2
        }{
            m!(n-2m)!
        }
        \leq 
        \mathcal{C}_{K,k}c_n\,,
    \end{align}
    where
    \begin{align}
        c_n
        \coloneqq
        \sum_{m=0}^{\lfloor n/2\rfloor}
        \frac{
            k^m l^{n-2m}
        }{
            m!(n-2m)!
        }\,.
        \label{eq: c_n}
    \end{align}

    Next, we use some basic concepts in Segal-Bargmann space \cite{Hall2013quantum}. The Segal-Bargmann space $\mathcal{H}L^2(\C^n,\mu)$ is a Hilbert space with inner product
    \begin{align}
        \braket{f,g}_\mu &= \frac{1}{\pi^n} \int_{\C^n} \!\dd^{2n}\bm{z}\,e^{-|\bm{z}|^2} \,\bar{f}(\bm{z})g(\bm{z})
    \end{align}
    For our purposes we only need $n=1$. Let
    \begin{align}
        f(z)\coloneqq e^{kz^2+lz}\qquad k,l\geq 0\,.
    \end{align}
    Then its squared-norm is given by
    \begin{align}
        ||f||_\mu^2 &= \frac{1}{\pi} \int_{\C} \!\dd^2z\,e^{-|z|^2}e^{k\bar{z}^2+l\bar{z}}e^{kz^2+lz}\notag\\
        &=  \frac{1}{\pi} \int_{\C} \!\dd^2z\,e^{-|z|^2}e^{k(z^2+\bar{z}^2)+2l \Re(z)}\notag\\
        &= \frac{1}{\pi}\int_{\R^2}\dd x\,\dd y\, e^{-x^2-y^2} e^{2 x (k x+l)-2 k y^2}\notag\\ 
        &= \frac{1}{\sqrt{1-4 k^2}} \exp\rr{\frac{l^2}{1-2 k}}\,,\quad k< \frac{1}{2}\,.
    \end{align}
    The last equality is valid for $k<1/2$. 

    Since $f(z)$ is holomorphic on $\C$, it can be written as a power series
    \begin{align}
        f(z) = \sum_{n=0}^\infty \sqrt{n!}c_n e_n(z)\,,\quad  
        e_n(z)\coloneqq \frac{z^n}{\sqrt{n!}}
    \end{align}
    where $c_n$ is given in Eq.~\eqref{eq: c_n} and $e_n(z)$ forms an orthonormal basis of the Segal--Bargmann space. Then for $k<1/2$ we have
    \begin{align}
        \sum_{n=0}^{\infty}
        n!|c_n|^2 = \lVert f\rVert_\mu^2
        <\infty\,.
    \end{align}
    Finally, the spectral norm satisfies
    \begin{align}
        \lVert\bar X\otimes Y\rVert_2
        =
        \lVert X\rVert_2\lVert Y\rVert_2\,.
    \end{align}
    Therefore
    \begin{align}
        \sum_{n=0}^{\infty}
        \lVert\bar{A^n}\otimes A^n\rVert_2
        &=
        \sum_{n=0}^{\infty}
        \lVert A^n\rVert_2^2
        \leq
        \mathcal{C}_{K,k}^2
        \lVert V\rVert_2^2
        \sum_{n=0}^{\infty}
        n!c_n^2
        <\infty
    \end{align}
    and the series defining the transfer matrix $E$ converges absolutely.
\end{proof}

\section{Bosonic transfer matrix calculus}
\label{appendix: derivation-transfer}

In this section we detail some derivations involved in the computation of transfer matrices in the main text. For convenience, we first recall some simple facts about bosonic ladder operators.
\begin{lemma}
    \label{lemma: CCR-calculus}
    Let $f$ be any analytic function. Then for any constant $\alpha\in \C$, we have
    \begin{align}
        e^{\alpha a}f(a^\dagger) = f(a^\dagger + \alpha) e^{\alpha a}\,.
    \end{align}
\end{lemma}
\begin{proof}
    We use the fact that $e^{\alpha a}$ is invertible and using the Baker-Campbell-Hausdorff (BCH) formula we get
    \begin{align*}
        e^{\alpha a}a^\dagger e^{-\alpha a} &= a^\dagger+ \alpha[a, a^\dagger] = a^\dagger+\alpha
    \end{align*}
    where we have used the CCR $[a,a^\dagger]=\openone$. Since $f$ is assumed to be analytic, we can use its power series expansion to obtain
    \begin{align*}
        e^{\alpha a}f(a^\dagger) e^{-\alpha a} 
        &= \sum_{m=0}^\infty c_m(e^{\alpha a}a^\dagger e^{-\alpha a})^m
        = f(a^\dagger+\alpha)
    \end{align*}
    and the result follows.
\end{proof}

\begin{corollary}
    \label{cor: push-through-operators}
    Let $f$ be an analytic matrix-valued function defined through its power series so that
    \begin{align}
        f(\openone\otimes a^\dagger) &\coloneqq \sum_{m=0}^\infty c_m\otimes (a^\dagger)^m 
    \end{align}
    where $c_m\in M_n(\C)$. If $[A,c_m]=0$ for all $m$ then
    \begin{align}
        e^{A\otimes a}f(\openone \otimes a^\dagger) = f(\openone \otimes a^\dagger + A\otimes \openone)e^{A\otimes a}\,.
    \end{align}
    
\end{corollary}

\begin{proof}
    We have
    \begin{align*}
        e^{A\otimes a}(\openone \otimes a^\dagger)e^{-A\otimes a} = \openone \otimes a^\dagger+ A\otimes \openone\,,
    \end{align*}
    so that the same steps as before leads to
    \begin{align*}
         e^{A\otimes a}f(\openone\otimes a^\dagger)e^{-A\otimes a} 
         &= \sum_{m=0}^\infty (c_m\otimes\openone)e^{A\otimes  a}(\openone\otimes a^\dagger)^m e^{-A\otimes a}\notag\\
         &= \sum_{m=0}^\infty (c_m\otimes \openone)(\openone \otimes a^\dagger+A\otimes \openone)^m\notag\\
         &\equiv f(\openone \otimes a^\dagger + A\otimes \openone)
    \end{align*}
    and the result follows.
\end{proof}
\noindent The corollary allows us to use a shorthand 
\begin{align*}
     e^{A\otimes a}f(a^\dagger) = f(a^\dagger + A)e^{A\otimes a}
\end{align*}
where no confusion should arise with regards to the tensor product factor.

Consider pairwise commuting matrices $P,Q,R,S\in M_n(\C)$ and define
\begin{align}
    \mathcal{E}
    \coloneqq
    \braket{
        0|
        e^{Pa}e^{Qa^2}
        e^{R(a^\dagger)^2}e^{Sa^\dagger}
        |0
    }
    \in M_n(\C)\,,
\end{align}
where we use the shorthand $Pa\equiv P\otimes a$, and similarly for
the other terms. We assume that the matrix-valued Gaussian integral
below converges. A sufficient condition is that there exists a
submultiplicative matrix norm for which
\begin{align}
    \lVert Q\rVert+\lVert R\rVert<1\,.
    \label{eq: Gaussian-sufficient-convergence}
\end{align}
Indeed, using $\lVert e^X\rVert\leq e^{\lVert X\rVert}$, the norm of
the integrand is bounded by
\begin{align*}
    \left|\left|e^{-|\alpha|^2}
    e^{
        Q\alpha^2
        +R\bar\alpha^2
        +(S+2RP)\bar\alpha
    }\right|\right|\leq e^{
        -\rr{1-\lVert Q\rVert-\lVert R\rVert}|\alpha|^2
        +C|\alpha|
    }
\end{align*}
for some constant $C>0$, which is integrable whenever
$\lVert Q\rVert+\lVert R\rVert<1$.

Using Corollary~\ref{cor: push-through-operators}, we obtain
\begin{align}
    \mathcal{E}
    &=
    \braket{
        0|
        e^{Qa^2}e^{Pa}
        e^{R(a^\dagger)^2}e^{Sa^\dagger}
        |0
    }
    \notag\\
    &=
    \braket{
        0|
        e^{Qa^2}
        e^{R(a^\dagger+P)^2}
        e^{S(a^\dagger+P)}
        |0
    }
    \notag\\
    &=
    e^{RP^2+SP}
    \braket{
        0|
        e^{Qa^2}
        e^{R(a^\dagger)^2}
        e^{(S+2RP)a^\dagger}
        |0
    }\,.
    \label{eq: Gaussian-shift-P}
\end{align}
Using the resolution of the identity in terms of normalized coherent
states,
\begin{align}
    \openone
    =
    \int_{\C}\frac{\dd^2\alpha}{\pi}
    \ketbra{\alpha}{\alpha}\,,
\end{align}
gives
\begin{align}
    &\braket{
        0|
        e^{Qa^2}
        e^{R(a^\dagger)^2}
        e^{(S+2RP)a^\dagger}
        |0
    }
    \notag\\
    &\qquad =
    \int_{\C}\frac{\dd^2\alpha}{\pi}
    e^{-|\alpha|^2}
    e^{
        Q\alpha^2
        +R\bar\alpha^2
        +(S+2RP)\bar\alpha
    }
    \notag\\
    &\qquad =
    \Delta^{\frac{1}{2}}
    e^{\Delta Q(S+2RP)^2}\,,
    \qquad
    \Delta\coloneqq(\openone-4QR)^{-1}\,.
\end{align}
Here $\Delta^{1/2}$ is the principal square root of $\Delta$. Therefore,
\begin{align}
    \mathcal{E} \equiv \mathcal{E}(P,Q,R,S) 
    &=
    \Delta^{\frac{1}{2}}
    e^{RP^2+SP+\Delta Q(S+2RP)^2}
    \notag\\
    &=
    \Delta^{\frac{1}{2}}
    e^{\Delta(RP^2+SP+QS^2)}\,.
    \label{eq: Gaussian-matrix}
\end{align}

To insert anti-normal-ordered monomials, we introduce two scalar
sources $s,t\in\C$:
\begin{align}
    \mathcal{E}(s,t)
    &\coloneqq
    \mathcal{E}
    \left(
        P+s\openone,
        Q,
        R,
        S+t\openone
    \right)
    \notag\\
    &=
    \braket{
        0|
        e^{(P+s\openone)a}
        e^{Qa^2}
        e^{R(a^\dagger)^2}
        e^{(S+t\openone)a^\dagger}
        |0
    }.
    \label{eq: source-dependent-Gaussian-kernel}
\end{align}
Since the source terms commute with the other annihilation or creation
operators in their respective exponentials, differentiation gives
\begin{align}
    \left.
    \partial_s^m\partial_t^n
    \mathcal{E}(s,t)
    \right|_{s,t=0}
    &=
    \braket{
        0|
        e^{Pa}e^{Qa^2}
        a^m(a^\dagger)^n
        e^{R(a^\dagger)^2}e^{Sa^\dagger}
        |0
    }\,.
    \label{eq: source-derivatives-Gaussian-kernel}
\end{align}

We can now connect these to the expression for transfer matrices \eqref{eq: transfer-matrix-general}. We claim that if $K,L$ are chosen to commute, then the transfer matrices can be computed explicitly in closed form as functions of $V,K,L$ as given in Proposition~\ref{prop: transfer-matrix}.
\begin{proof}[Proof of Proposition~\ref{prop: transfer-matrix}]
    We start from the MPS tensor
    \begin{align*}
        \sum_{n=0}^{\infty}\bar{A^n}\otimes A^n
        =
        (\bar V\otimes V)
        \sum_{n=0}^{\infty}
        n!\,\bar C_n\otimes C_n.
    \end{align*}
    where 
    \begin{align*}
        C_n = \sum_{m=0}^{\lfloor n/2\rfloor} \frac{K^mL^{n-2m}}{m!(n-2m)!}
    \end{align*}
    Observe that 
    \begin{align*}
        e^{K\otimes (a^\dagger)^2}e^{L\otimes a^\dagger}\ket{0} 
        &= \sum_{m,n=0}^\infty \frac{K^mL^n}{m!n!}\sqrt{(2m+n)!}\ket{2m+n} \notag\\
        &\equiv \sum_{n=0}^\infty\sqrt{n!}\sum_{m=0}^{\lfloor n/2\rfloor} \frac{K^mL^{n-2m}}{m!(n-2m)!}\ket{n} \notag\\
        &\equiv \sum_{n=0}^\infty \sqrt{n!}C_n\otimes \ket{n}\,.
    \end{align*}
    Consequently, we identify the MPS tensor as
    \begin{align*}
        A = Ve^{K\otimes (a^\dagger)^2}e^{L\otimes a^\dagger}\ket{0} \,.
    \end{align*}

    Define a `star' product as a shorthand for  
    \begin{align}
        (A\otimes B)\star(C\otimes D) \coloneqq (A\otimes C)\otimes BD\,,
    \end{align}
    i.e., tensor product on first tensor factor and standard operator product on the second tensor factor: this occurs, for instance, when one multiplies two matrix-product operators. The transfer matrix can then be written as
    \begin{align*}
        E &=  \bar V \otimes{V}\bra{0}e^{\bar L\otimes a}e^{\bar K\otimes a^2}\star e^{{K}\otimes (a^\dagger)^2}e^{{L}\otimes a^\dagger}\ket{0} \notag\\
        &=  \bar V \otimes{V} \mathcal{E}  \in M_{D^2}(\C)\,,
    \end{align*}
    where the matrices $P,Q,R,S$ in $\mathcal{E}$ are chosen to be
    \begin{equation}
        \begin{aligned}
            P &=  \bar{L}\otimes \openone\,,\quad Q = \bar K\otimes \openone\,,\\
            R &= {\openone}\otimes {K}\,,\quad S =\openone \otimes {L}\,.
        \end{aligned}
        \label{eq: PQRS}
    \end{equation}
    For these choices, $P,Q,R,S$ commute pairwise and Eq.~\eqref{eq: Gaussian-matrix} gives the required closed-form expression.
\end{proof}
For the BMPS specialization~\eqref{eq: PQRS}, convergence follows
directly from Appendix~\ref{appendix: ansatz-absolute-convergence}.
Indeed, when $\rho(K)<1/2$, the transfer-matrix series converges
absolutely, and the coherent-state resolution of the identity used above
is justified entrywise by the Cauchy--Schwarz inequality. Moreover,
\begin{align}
    \rho(4QR) 
    = 4\rho(\bar K\otimes K) 
    = 4\rho(K)^2 < 1\,,
\end{align}
so that $\Delta=(\openone-4QR)^{-1}$ and its principal square root are
well-defined and Eq.~\eqref{eq: Gaussian-matrix} is valid.

\begin{proof}[Proof of Proposition~\ref{prop: differentiate-E}]
By Eq.~\eqref{eq: source-dependent-Gaussian-kernel}, the scalar
sources enter only through the first and fourth exponentials. Since
$a$ commutes with $e^{\bar L a}$ and $e^{\bar K a^2}$, while
$a^\dagger$ commutes with $e^{K(a^\dagger)^2}$ and
$e^{L a^\dagger}$, differentiation gives
\begin{align}
    &
    \partial_s^m\partial_t^n E(s,t)
    \bigr|_{s,t=0}\notag\\
    &=
    (\bar V\otimes V)
    \braket{
        0|
        e^{\bar L a}
        e^{\bar K a^2}
        a^m(a^\dagger)^n
        e^{K(a^\dagger)^2}
        e^{L a^\dagger}
        |0
    }
    \notag\\
    &=
    E_{a^m(a^\dagger)^n}.
\end{align}
The derivatives are justified by the absolute convergence of the
source-dependent expression in a neighbourhood of $s=t=0$. The
sources modify only the linear terms and hence do not change the
condition $\rho(K)<1/2$.
\end{proof}

The transfer matrix can also be computed efficiently for certain non-polynomial functions of $a,a^\dagger$, notably when the local operator $O$ is a Gaussian unitary, which is the statement of Proposition~\ref{prop: transfer-matrix-exponential-operators}.

\begin{proof}[Proof of Proposition~\ref{prop: transfer-matrix-exponential-operators}]
Write $\zeta=re^{i\phi}$ and define
\begin{align}
    \tau\coloneqq e^{i\phi}\tanh r,
    \qquad
    c\coloneqq\cosh r.
\end{align}
The displacement and squeezing operators admit the anti-normal-ordered decompositions
\begin{subequations}
\begin{align}
    D(\alpha)
    &=
    e^{|\alpha|^2/2}
    e^{-\bar\alpha a}
    e^{\alpha a^\dagger}\,,
    \label{eq: displacement-antinormal}
    \\
    S(\zeta)
    &=
    \sqrt{c}\,
    e^{\bar\tau a^2/2}
    c^n
    e^{-\tau(a^\dagger)^2/2}\,,
    \qquad
    n\coloneqq a^\dagger a\,.
    \label{eq: squeezing-antinormal}
\end{align}
\end{subequations}
We will use the identities
\begin{subequations}
\begin{align}
    e^{\alpha a^\dagger}f(a)
    &=
    f(a-\alpha)e^{\alpha a^\dagger}\,,
    \label{eq: creation-push-annihilation}
    \\
    e^{\alpha a^\dagger}c^n
    &=
    c^n e^{\alpha a^\dagger/c}\,,
    \label{eq: creation-push-number}
    \\
    f(a^\dagger)e^{-i\theta n}
    &=
    e^{-i\theta n}f(e^{i\theta}a^\dagger)\,,
    \label{eq: creation-push-rotation}
\end{align}
\end{subequations}
which follow from
\begin{align}
    e^{\alpha a^\dagger}ae^{-\alpha a^\dagger}
    &=
    a-\alpha\,,
    &
    c^na^\dagger c^{-n}
    &=
    ca^\dagger\,.
\end{align}

Using Eqs.~\eqref{eq: displacement-antinormal}--\eqref{eq: creation-push-rotation}, the Gaussian unitary
\begin{align}
    U_{\mathsf G}
    =
    D(\alpha)S(\zeta)e^{-i\theta n}
\end{align}
can be written as
\begin{align}
    U_{\mathsf G}
    &=
    \sqrt{c}\,
    e^{\frac{1}{2}
        \rr{
            |\alpha|^2+\bar\tau\alpha^2
        }}
    e^{
        \frac{1}{2}\bar\tau a^2
        -
        (\bar\alpha+\bar\tau\alpha)a
    }
    \rr{ce^{-i\theta}}^n
    \notag\\
    &\quad\times
    e^{
        -\frac{1}{2}\tau e^{2i\theta}(a^\dagger)^2
        +
        \frac{\alpha e^{i\theta}}{c}a^\dagger
    }.
    \label{eq: general-Gaussian-antinormal}
\end{align}

Recall the matrix-valued function $\mathcal{E}$ introduced in
Eq.~\eqref{eq: Gaussian-matrix},
\begin{align}
    \mathcal{E}(P,Q,R,S)
    \coloneqq
    \braket{
        0|
        e^{Pa}e^{Qa^2}
        e^{R(a^\dagger)^2}e^{Sa^\dagger}
        |0
    },
    \label{eq: Gaussian-kernel-shorthand}
\end{align}
where $P,Q,R,S$ commute pairwise. Observe that if we insert $U_{\mathsf G}$ in the middle, i.e.,
\begin{align}
    \mathcal{E}_{U_{\mathsf G}}
    \coloneqq
    \braket{
        0|
        e^{Pa}e^{Qa^2}
        U_{\mathsf G}
        e^{R(a^\dagger)^2}e^{Sa^\dagger}
        |0
    }\,,
\end{align}
the first exponential in Eq.~\eqref{eq: general-Gaussian-antinormal} can be absorbed by the replacement
\begin{align}
    P
    &\mapsto
    P'
    \coloneqq
    P-(\bar\alpha+\bar\tau\alpha)\openone,
    \notag\\
    Q
    &\mapsto
    Q'
    \coloneqq
    Q+\frac{1}{2}\bar\tau\openone.
    \label{eq: Gaussian-left-substitution}
\end{align}
Next, set $q\coloneqq ce^{-i\theta}$. Then for every analytic function $f$ of $a^\dagger$,
\begin{align}
    q^nf(a^\dagger)\ket{0}
    =
    f(qa^\dagger)\ket{0},
    \label{eq: number-scaling-vacuum}
\end{align}
since $q^na^\dagger q^{-n}=qa^\dagger$ and
$q^n\ket{0}=\ket{0}$. Consequently, the remaining factors are absorbed by the replacement
\begin{align}
    R
    &\mapsto
    R'
    \coloneqq
    c^2\left(
        e^{-2i\theta}R
        -
        \frac{1}{2}\tau\openone
    \right),
    \notag\\
    S
    &\mapsto
    S'
    \coloneqq
    ce^{-i\theta}S+\alpha\openone.
    \label{eq: Gaussian-right-substitution}
\end{align}
We therefore obtain
\begin{align}
    \mathcal{E}_{U_{\mathsf G}}
    &=
    \sqrt{c}\,
    e^{\frac{1}{2}
        \rr{
            |\alpha|^2+\bar\tau\alpha^2
        }}
    \mathcal{E}(P',Q',R',S')\,.
    \label{eq: Gaussian-unitary-kernel}
\end{align}
Substituting $P,Q,R,S$ with Eq.~\eqref{eq: PQRS} into Eqs.~\eqref{eq: Gaussian-left-substitution} and \eqref{eq: Gaussian-right-substitution}, we get
\begin{align}
    \bar L\otimes\openone
    &\mapsto
    \left[
        \bar L
        -
        (\bar\alpha+\bar\tau\alpha)\openone
    \right]\otimes\openone,
    \notag\\
    \bar K\otimes\openone
    &\mapsto
    \left[
        \bar K+\frac{1}{2}\bar\tau\openone
    \right]\otimes\openone,
    \notag\\
    \openone\otimes K
    &\mapsto
    \openone\otimes
    c^2\left(
        e^{-2i\theta}K
        -
        \frac{1}{2}\tau\openone
    \right),
    \notag\\
    \openone\otimes L
    &\mapsto
    \openone\otimes
    \left(
        ce^{-i\theta}L+\alpha\openone
    \right).
\end{align}
Substituting these expressions into Eq.~\eqref{eq: transfer-matrix-bosons} gives the result.
\end{proof}

\section{Equivalent sets of BMPS}
	\label{appendix: set-inclusion}

    Here we quote the following statement from Proposition~\ref{prop: CV-MPS-equivalence} for convenience: we want to show that
    \begin{enumerate}[leftmargin=*,label={(\roman*)}]
        \item Every $\ket{\Psi_N^{\exp}(D)}$ can be expressed as $\ket{\Psi_N^{\mathsf{PAG}}(\mathsf{d},D)}$ with $\mathsf{d} \leq 2D^2-D$. 
			
        \item Conversely, every $\ket{\Psi_N^{\mathsf{PAG}}(\mathsf{d},D)}$ can be expressed as $\ket{\Psi_N^{\exp}(D\mathsf{d})}$ with commuting generators.
    \end{enumerate}
	
	\begin{proof}[Proof of Proposition~\ref{prop: CV-MPS-equivalence}]
		(i) We first show that the exponential ansatz is an MPS of photon-added
		Gaussian states. Let
		\begin{align}
			K
			&=
			\sum_{r=1}^{r_{\max}}
			\bigl(\kappa_r P_r+\mathcal R_r\bigr),
			&
			L
			&=
			\sum_{s=1}^{s_{\max}}
			\bigl(\ell_s Q_s+\mathcal N_s\bigr)
		\end{align}
		be the Jordan decompositions of $K$ and $L$, where $P_r$ and $Q_s$
		are the respective spectral projectors and
		\begin{align}
			\mathcal R_r
			&=
			(K-\kappa_r\openone)P_r,
			&
			\mathcal N_s
			&=
			(L-\ell_s\openone)Q_s
		\end{align}
		are their nilpotent parts. Then 
		\begin{equation}
			\begin{aligned}
				e^{K\otimes (a^\dagger)^2}
				&=
				\sum_{r=1}^{r_{\max}}
				\sum_{p=0}^{\rank{P_r}-1}
				\frac{\mathcal R_r^pP_r}{p!}\,
				\otimes (a^\dagger)^{2p}e^{\kappa_r (a^\dagger)^2}\,,
				\\
				e^{L\otimes a^\dagger}
				&=
				\sum_{s=1}^{s_{\max}}
				\sum_{q=0}^{\rank{Q_s}-1}
				\frac{\mathcal N_s^qQ_s}{q!}\,
				\otimes (a^\dagger)^q e^{\ell_s a^\dagger}\,.
			\end{aligned}
		\end{equation}
		Hence we can write
		\begin{align}
			V e^{K(a^\dagger)^2}e^{La^\dagger}\ket{0}
			&=
			\sum_{n,r,s}
			\mathsf{A}^{n,r,s}
			\ket{n,\kappa_r,\ell_s}\,,
			\label{eq:exp-local-photon-expansion}
		\end{align}
		where $\ket{n,\kappa_r,\ell_s}$ is the photon-added Gaussian state \eqref{eq: photon-added-Gaussian-states} and 
		\begin{align}
			\mathsf{A}^{n,r,s} &\coloneqq
			V\sum_{p,q}
			\frac{
				\mathcal R_r^pP_r
				\mathcal N_s^qQ_s
			}{p!q!} \in M_D(\C)
		\end{align}
		and the summation is over $0\leq p\leq \rank(P_r)-1, 0\leq q\leq \rank(Q_s)-1$, and $2p+q=n$. For each pair $(r,s)$, at most $(2\rank{P_r}+\rank{Q_s}-2)$ local states occur.
		Therefore, using the fact that $\sum_r \rank(P_r) = \sum_s \rank(Q_s) = D$ and that $1\leq r_{\max},s_{\max}\leq D$, we obtain
		\begin{align*}
			\mathsf d
			&\leq
			\sum_{r=1}^{r_{\max}}
			\sum_{s=1}^{s_{\max}}
			\bigl(2\rank{P_r}+\rank{Q_s}-2\bigr)
			\notag\\
            &=
			2Ds_{\max}+Dr_{\max}-2r_{\max}s_{\max}
			\notag\\
            &\leq 2D^2-D\,.
		\end{align*}
		Substituting Eq.~\eqref{eq:exp-local-photon-expansion} at every site
		and contracting the virtual indices gives
		\begin{align}
			\ket{\Psi_N^{\exp}(D)}
			=
			\sum_{I_1,\ldots,I_N=1}^{\mathsf d}
			\Tr_D\rr{
				B\mathsf{A}^{I_1}\cdots\mathsf{A}^{I_N}
			}
			\ket{I_1\cdots I_N},
		\end{align}
		where $I=(n,r,s)$ and terms for which
		$\mathsf{A}^{n,r,s}=0$ may be omitted. Hence every
		$\ket{\Psi^{\exp}_N(D)}$ can be expressed as
		$\ket{\Psi^{\mathsf{PAG}}_N(\mathsf{d},D)}$ with
		$\mathsf d\leq 2D^2-D$. 
		
		(ii) Next, we show that the MPS of photon-added Gaussian states can be expressed in the form of a commuting-generator exponential ansatz. We are given $\ket{\Psi^{\mathsf{PAG}}_N(\mathsf{d},D)}$ with tensors $\mathsf{A}^I$, where $I=(q,\lambda)$ labels the photon-added Gaussian states $\ket{q,\kappa_\lambda,\ell_\lambda}$. First we define a $\mathsf{d}$-dimensional auxiliary space $\mathcal{H}_{\mathsf{aux}}$ given by
		\begin{align}
			\mathcal H_{\mathrm{aux}}
			=
			\bigoplus_\lambda
			\Span\{
			\ket{e_{0,\lambda}},\ldots,
			\ket{e_{\nu_\lambda,\lambda}}
			\}\cong \C^{\mathsf{d}} \,,
		\end{align}
		where the vectors $\{\ket{e_{q,\lambda}}\}$ form an orthonormal
		basis and we write $\ket{e_I}=\ket{e_{q,\lambda}}$. Then
		consider the weighted shift operator on $\mathcal{H}_{\mathsf{aux}}$
		\begin{align}
			\mathfrak N_\lambda
			=
			\sum_{q=1}^{\nu_\lambda}
			q\ketbra{e_{q,\lambda}}{e_{q-1,\lambda}}\,,\quad \mathfrak N_\lambda^q\ket{e_{0,\lambda}}
			=
			q!\ket{e_{q,\lambda}}\,.
		\end{align}
        This will account for the photon addition, which arises from the nilpotent part of the matrix in the exponential ansatz. 
        
		Now consider a virtual space $\C^D\otimes \mathcal H_{\mathrm{aux}}$ for the exponential ansatz $\ket{\Psi_{N}^{\exp}(D\mathsf{d})}$ and set
		\begin{equation}
			\begin{aligned}
				\mathsf{B}
				&= B\otimes\openone_{\mathsf d}\,,\\
				\mathsf{V}
				&= \sum_I \mathsf{A}^I\otimes \ketbra{\sigma}{e_I}\,, \qquad 
				\ket{\sigma}
				=\sum_\lambda\ket{e_{0,\lambda}}\,,\\
				\mathsf{K}
				&=\openone_D\otimes  \mathsf{K}_{\mathsf{aux}}\,,\quad  \mathsf{K}_{\mathsf{aux}}\coloneqq \bigoplus_\lambda
				\kappa_\lambda\openone_{\nu_\lambda+1}\,,\\
				\mathsf{L}
				&=\openone_D\otimes  \mathsf{L}_{\mathsf{aux}}\,,\quad  \mathsf{L}_{\mathsf{aux}}\coloneqq \bigoplus_\lambda
				\left(
				\ell_\lambda\openone_{\nu_\lambda+1}
				+\mathfrak N_\lambda
				\right)\,.
			\end{aligned}
            \label{eq: exponential-embedding}
		\end{equation}
        Observe that $[\mathsf{K},\mathsf{L}]=0$ because both operators are
		block diagonal in $\lambda$ and $\mathsf{K}$ is scalar on each
		block. Furthermore, for $I=(q,\lambda)$,
        \begin{align}
            \bra{e_I}
            e^{\mathsf{K}_{\mathsf{aux}}\otimes(a^\dagger)^2}
            e^{\mathsf{L}_{\mathsf{aux}}\otimes a^\dagger}
            \ket{\sigma}\ket{0}
            =
            \ket{I},
            \label{eq: auxiliary-embedding-identity}
        \end{align}
        Here $\ket I=\ket{q,\kappa_\lambda,\ell_\lambda}$. To apply this identity to the full ansatz, let
        \begin{align}
            G_n
            =
            e^{\mathsf{K}_{\mathsf{aux}}\otimes(a_n^\dagger)^2}
            e^{\mathsf{L}_{\mathsf{aux}}\otimes a_n^\dagger}\,.
        \end{align}
        To verify that indeed this works, we expand $\mathsf V$ given in Eq.~\eqref{eq: exponential-embedding}:
        \begin{align}
            &\ket{\Psi_N^{\exp}(D\mathsf d)}
            \notag\\
            &=
            \sum_{J_1,\ldots,J_N=1}^{\mathsf d}\!\!\!\!
            \Tr_D\rr{
                B\mathsf{A}^{J_1}\cdots\mathsf{A}^{J_N}
            }
            \Tr_{\mathsf d}\rr{
                \prod_{n=1}^N
                \ketbra{\sigma}{e_{J_n}}\mathsf G_n
            }
            \ket{0}^{\otimes N}
            \notag\\
            &=
            \sum_{J_1,\ldots,J_N=1}^{\mathsf d}\!\!\!\!
            \Tr_D\rr{
                B\mathsf{A}^{J_1}\cdots\mathsf{A}^{J_N}
            }
            \bigotimes_{n=1}^N
                \rr{\bra{e_{J_n}}\mathsf G_n\ket{\sigma}}\ket{0}
            \notag\\
            &=
            \sum_{J_1,\ldots,J_N=1}^{\mathsf d}\!\!\!\!
            \Tr_D\rr{
                B\mathsf{A}^{J_1}\cdots\mathsf{A}^{J_N}
            }
            \ket{J_1\cdots J_N}
            \notag\\
            &=
            \ket{\Psi_N^{\mathsf{PAG}}(\mathsf d,D)}.
        \end{align}
        The third follows from Eq.~\eqref{eq: auxiliary-embedding-identity}.
		
	\end{proof}

    \section{Parent Hamiltonians: general constructions}
    \label{appendix: parent-details}
    
    We give a single-mode construction for any finite family of
    photon-added squeezed coherent states. In the MPS application, the
    relevant single-mode vectors have the form
    \begin{align}
        \ket{\psi_\mu}
        =
        \sum_{j=1}^{n}
        \Tr_D\left[
            X_\mu\mathsf A^j
        \right]
        \ket{e_j}\,,
    \end{align}
    for suitable matrices $X_\mu\in M_D(\C)$. We consider the slightly more
    general setting where for all $j=1,2,\ldots,n$ we define
    \begin{align}
        \ket{e_j}
        &=
        f_j(a^\dagger)\ket{0}\,,
        \quad 
        f_j(z)
        \coloneqq
        u_j(z)e^{\kappa_jz^2+\ell_jz}\,,
        \label{eq: general-squeezed-basis}
    \end{align}
    where $u_j$ is a polynomial and $|\kappa_j|<1/2$. We assume that
    $\ket{e_1},\ldots,\ket{e_n}$ are linearly independent and define
    \begin{align}
        \mathcal V
        \coloneqq
        \Span_\C
        \left\{
            \ket{e_j}:j=1,\ldots,n
        \right\}.
    \end{align}
    
    Let
    \begin{align}
        \ket{\psi_\mu}
        =
        \sum_{j=1}^{n}
        C_{\mu j}\ket{e_j}\,,
        \qquad
        \mu=1,\ldots,m\,,
    \end{align}
    and suppose that the coefficient matrix
    $C\in M_{m,n}(\C)$ has rank $m$. The target space is
    \begin{align}
        \mathcal W
        \coloneqq
        \Span_\C
        \left\{
            \ket{\psi_\mu}:
            \mu=1,\ldots,m
        \right\}
        \subseteq
        \mathcal V\,.
    \end{align}
    Choose a matrix $R\in M_{n-m,n}(\C)$ of rank $n-m$ such that $RC^T=0$. Our task is to construct a parent Hamiltonian $h$ with $\ker h=\mathcal{W}$ (cf. Sec.~\ref{subsec: abstract-parent}). 
    
    To set things up, we first define the polynomials $p_{rj}$ by
    \begin{align}
        f_j^{(r)}(z)
        =
        p_{rj}(z)e^{\kappa_jz^2+\ell_jz}\,.
        \label{eq: squeezed-derivative-polynomials}
    \end{align}
    For the functions $f_j$ above, repeated use of the CCR gives
    \begin{align}
        a^r\ket{e_j}
        =
        p_{rj}(a^\dagger)
        e^{\kappa_j(a^\dagger)^2+\ell_j a^\dagger}
        \ket{0}\,.
        \label{eq: squeezed-derivative-action}
    \end{align}

    We first prove that we can construct the positive operator $h_0=Q^\dagger Q$ whose kernel is the ambient space $\mathcal{V}$. At a high level, the idea behind this is to view $Q\ket{e_j}=0$ as an ordinary differential equations (ODEs): 
    \begin{align}
        Q \ket{e_j} = 0\quad  \Longrightarrow \quad \hat{Q}f_j(z) = 0\,,
    \end{align}
    i.e., we are looking for a differential operator $\hat{Q}(z,\partial_z)$ such that it annihilates every linearly independent basis functions $f_j$. Then by replacing $z\mapsto a^\dagger$ and $\partial_z\mapsto a$, we obtain the desired operator $Q$. For this reason, standard techniques from ODE such as the Wronskian will be useful to construct the differential operator for which the basis functions $f_j$ span its solution space. 
    \begin{lemma}[Ambient space annihilator]
    \label{lemma: squeezed-Wronskian-annihilator}

    Define
    \begin{equation}
        \begin{aligned}
            \mathsf{P}(z)
            &\coloneqq
            \left[
                p_{rj}(z)
            \right]_{
                \substack{
                    r=0,\ldots,n-1\\
                    j=1,\ldots,n
                }
            }\,,\\ 
            \widetilde{\mathsf{P}}(z)
            & \coloneqq
            \left[
                p_{rj}(z)
            \right]_{
                \substack{
                    r=0,\ldots,n\\
                    j=1,\ldots,n
                }
            }\,.
        \label{eq: squeezed-Wronskian-polynomial}
        \end{aligned}
    \end{equation}
    and $\widetilde{\mathsf{P}}_{[r]}(z)$ is obtained by deleting
    the $r$-th row of $\tilde{P}(z)$. Then $\Delta(z)\coloneqq\det\mathsf{P}(z)$ is non-zero polynomial and the operator
    \begin{align}
        Q
        &\coloneqq
        \sum_{r=0}^{n}
        q_r(a^\dagger)a^r\,,
        &
        q_r(z)
        &\coloneqq
        (-1)^{n+r}
        \det\widetilde{\mathsf{P}}_{[r]}(z)
        \label{eq: squeezed-annihilator}
    \end{align}
    satisfies $\ker Q=\mathcal V$. 
\end{lemma}

\begin{proof}
    We seek a normally ordered operator
    \begin{align}
        Q
        &\coloneqq
        \sum_{r=0}^{n}
        q_r(a^\dagger)a^r
    \end{align}
    that annihilates each $\ket{e_j}$. Using Eq.~\eqref{eq: squeezed-derivative-action} we find that $Q\ket{e_j}=0$ provided
    \begin{align}
        \sum_{r=0}^{n}
        q_r(z)p_{rj}(z)
        =
        0\,,
        \qquad
        j=1,\ldots,n\,.
        \label{eq: squeezed-left-kernel}
    \end{align}
    Thus the row vector $\bm{q}(z)\coloneqq (q_0(z),\ldots,q_n(z))$ must lie in the left kernel of $\widetilde{\mathsf{P}}(z)$. 

    For each $j$, let
    \begin{align}
        \bm v^{[j]}(z)
        \coloneqq
        \bigl(
        p_{0j}(z),\ldots,p_{nj}(z)
        \bigr)^T
    \end{align}
    be the $j$th column of $\widetilde{\mathsf P}(z)$, and define
    \begin{align}
        \mathsf M_j(z)
        \coloneqq
        \left[
        \widetilde{\mathsf P}(z)
        \,\middle|\,
        \bm v^{[j]}(z)
        \right].
    \end{align}
    Since the last column of $\mathsf M_j(z)$ duplicates its $j$th
    column, $\det\mathsf M_j(z)=0$. Using cofactor expansion along the last column of $\mathsf{M}_j(z)$, we have
    \begin{align}
        0
        &=
        \det\mathsf M_j(z)
        \notag\\
        &=
        \sum_{r=0}^{n}
        (-1)^{n+r}
        p_{rj}(z)
        \det\widetilde{\mathsf P}_{[r]}(z)
        \notag\\
        &=
        \sum_{r=0}^{n}
        q_r(z)p_{rj}(z)\,.
        \label{eq: squeezed-cofactor-identity}
    \end{align}
               
    For $r=n$, deleting the last row of
    $\widetilde{\mathsf{P}}(z)$ gives $\mathsf{P}(z)$, and hence
    \begin{align}
        q_n(z)
        =
        \det\mathsf{P}(z)
        =
        \Delta(z)\,.
    \end{align}
    Furthermore, the Wronskian of $f_1,\ldots,f_n$ is
    \begin{align}
        \det
        \left[
            f_j^{(r)}(z)
        \right]_{
            \substack{
                r=0,\ldots,n-1\\
                j=1,\ldots,n
            }
        }
        =
        e^{\sum_{j=1}^{n}(\kappa_jz^2+\ell_jz)}
        \Delta(z)\,.
    \end{align}
    Since the $f_j$ are linearly independent analytic functions, their Wronskian is not identically zero. Therefore, $\Delta\not\equiv 0$ (not a zero polynomial) and in particular $Q$ contains powers of $a$ up to $a^n$. Eq.~\eqref{eq: squeezed-cofactor-identity} gives
    \begin{align}
        Q\ket{e_j}
        &=
        e^{\kappa_j(a^\dagger)^2+\ell_j a^\dagger}
        \sum_{r=0}^{n}
        q_r(a^\dagger)p_{rj}(a^\dagger)\ket{0}
        =
        0\,,
    \end{align}
    Hence $\mathcal V\subseteq\ker Q$. 
    
    Now choose $z_0$ such that $\Delta(z_0)\neq0$. In the Bargmann
    representation, the differential equation associated with $Q$ (replacing $a\mapsto \partial_z$ and $a^\dagger\mapsto z$) has order
    $n$ in a neighbourhood of $z_0$, and each solution is uniquely determined
    there by
    \begin{align}
        g(z_0),g'(z_0),\ldots,g^{(n-1)}(z_0)\,.
    \end{align}
    Consequently, $\dim\ker Q\leq n$ and since $\dim\mathcal V=n$, we obtain $\ker Q=\mathcal V$, as required.
\end{proof}

Next, we prove that the operator $F_j$ such that $F_j\ket{e_k} = \delta_{jk}\ket{0}$ can be systematically constructed. At a high level, the idea is that if we want $F_j\ket{e_k}=\delta_{jk}\ket{0}$, then we expect $F_j$ to take the form
\begin{align}
    F_j \propto a^d e^{-\kappa_j(a^\dagger)^2-\ell_j a^\dagger}G_j\,,
\end{align}
where $G_j$ is some auxiliary factor that isolates the $j$-th Gaussian components $(\kappa_j,\ell_j)$. More precisely, we require
\begin{align}
    G_j\ket{e_k}
    =
    \delta_{jk}\,
    \varphi(a^\dagger)
    e^{\kappa_k(a^\dagger)^2+\ell_k a^\dagger}
    \ket{0}\,,
    \label{eq: auxiliary-G-action}
\end{align}
for some polynomial $\varphi$ independent of $j$ and $k$. For this we adopt an analogous ansatz as $Q$, namely
\begin{align}
    G_j
    =
    \sum_{r=0}^{n-1}
    g_{jr}(a^\dagger)a^r
\end{align}
for some $g_{jr}(z)$ to be determined. Let $\mathsf G(z)=[g_{jr}(z)]$. Using
Eq.~\eqref{eq: squeezed-derivative-action}, we get 
\begin{align}
    G_j\ket{e_k} &= \sum_{r}g_{jr}(a^\dagger)p_{rk}(a^\dagger)e^{\kappa_k (a^\dagger)^2+\ell_k a^\dagger}\ket{0} \notag\\ 
    &\equiv [\mathsf{G}(a^\dagger)\mathsf{P}(a^\dagger)]_{jk}e^{\kappa_k (a^\dagger)^2+\ell_k a^\dagger}\ket{0} \,,
\end{align}
hence it boils down to the question of whether there exists $\mathsf{G}(z)$ such that $\mathsf{G}(z)\mathsf{P}(z) = \varphi(z)\openone_n$ where $\varphi(z)$ is some scalar function. The positive answer gives the required check operator.

\begin{lemma}
    \label{lemma: squeezed-extraction-operators}
    Write
    \begin{align}
        \Delta(z)
        =
        \gamma z^d+\text{lower-order terms},
        \qquad
        \gamma\neq0\,,
    \end{align}
    where $d$ is the degree of $\Delta(z)$ and define
    \begin{align}
        G_j
        &\coloneqq
        \sum_{r=0}^{n-1}
        \left[
            \operatorname{adj}\mathsf{P}(a^\dagger)
        \right]_{j,r+1}
        a^r\,,
        \label{eq: squeezed-intermediate-extractor}
        \\
        F_j
        &\coloneqq
        \frac{1}{\gamma d!}\,
        a^d
        e^{-\kappa_j(a^\dagger)^2-\ell_j a^\dagger}
        G_j\,.
        \label{eq: squeezed-extractor}
    \end{align}
    where $\operatorname{adj}(A)$ is the adjugate matrix of $A$, i.e., $A\operatorname{adj}(A)=\det(A)\openone$. Then $F_j\ket{e_k}=\delta_{jk}\ket{0}$.
\end{lemma}

\begin{proof}
    The required matrix $\mathsf{G}$ is in fact the adjugate matrix of $\mathsf{P}$ \cite{Horn1985matrixanalysis}. Indeed, the adjugate identity
    \begin{align}
        \operatorname{adj}\mathsf{P}(z)\mathsf{P}(z)
        =
        \Delta(z)\openone_n
    \end{align}
    gives
    \begin{align}
        G_j\ket{e_k}
        &=
        \sum_{r=0}^{n-1}
        \left[
        \operatorname{adj}\mathsf P(a^\dagger)
        \right]_{j,r+1}
        p_{rk}(a^\dagger)
        e^{\kappa_k(a^\dagger)^2+\ell_k a^\dagger}
        \ket{0}
        \notag\\
        &=
        \left[
        \operatorname{adj}\mathsf P(a^\dagger)
        \mathsf P(a^\dagger)
        \right]_{jk}
        e^{\kappa_k(a^\dagger)^2+\ell_k a^\dagger}
        \ket{0}
        \notag\\
        &=
        \delta_{jk}\Delta(a^\dagger)
        e^{\kappa_k(a^\dagger)^2+\ell_k a^\dagger}
        \ket{0}\,.
    \end{align}
    After canceling the squeezed-coherent part by $e^{-\kappa_j(a^\dagger)^2-\ell_j a^\dagger}$, the annihilation operator $a^d$ acts on $\Delta(a^\dagger)\ket{0}$ to give
    \begin{align}
        a^d\Delta(a^\dagger)\ket{0} = \gamma d!\ket{0}\,,
    \end{align}
    which fixes the normalization of $F_j$, hence $F_j\ket{e_k} = \delta_{jk}\ket{0}$ as required.
\end{proof}

Using the abstract construction in Sec.~\ref{subsec: abstract-parent}, we obtain the parent Hamiltonian.
\begin{proposition}[Single-mode parent term for photon-added squeezed states]
    \label{prop: squeezed-parent-construction}

    Let $Q$ and $F_j$ be the operators constructed in
    Lemmas~\ref{lemma: squeezed-Wronskian-annihilator}
    and~\ref{lemma: squeezed-extraction-operators}. Define
    \begin{equation}
        \begin{aligned}
            h_0
            &\coloneqq
            Q^\dagger Q\,,
            \\
            h_R
            &\coloneqq
            \sum_{\rho=1}^{n-m}
            O_\rho^\dagger O_\rho\,,\quad 
            O_\rho
            \coloneqq
            \sum_{j=1}^{n}
            R_{\rho j}F_j\,.
            \label{eq: squeezed-check-term}
        \end{aligned}
    \end{equation}
    Then $\ker h_0 = \mathcal V$ and $\ker(h_0+h_R)=\mathcal{W}$. 
\end{proposition}

\begin{proof}
    Lemma~\ref{lemma: squeezed-Wronskian-annihilator} gives $\ker h_0=\ker Q=\mathcal V$. For
    \begin{align}
        \ket{\phi}
        =
        \sum_{j=1}^{n}x_j\ket{e_j}
        \in\mathcal V\,,
    \end{align}
    Lemma~\ref{lemma: squeezed-extraction-operators} gives
    \begin{align}
        O_\rho\ket{\phi}
        =
        (Rx)_\rho\ket{0}\,.
    \end{align}
    Since $R$ has rank $n-m$ and satisfies $RC^T=0$, we have $\ker R=\im C^T$ and hence $\ker h_R\cap\mathcal V=\mathcal W$. Positivity of $h_0$ and $h_R$ gives $\ker(h_0+h_R)=\ker h_0\cap\ker h_R=\mathcal W$.
\end{proof}

\end{document}